\pdfoutput=1

\documentclass{manuscript}

\usepackage{times}
\usepackage{scrtime}
\usepackage[ruled,vlined,linesnumbered]{algorithm2e}
\usepackage{subfig}
\usepackage{epsfig}
\usepackage{float}
\usepackage{url}
\usepackage{color}
\usepackage{graphicx}
\usepackage{balance}
\usepackage{hyperref}
\usepackage{cleveref}

\newcommand{\Paragraph}[1]{\smallskip\noindent{\bf #1.}}

\definecolor{mybg}{rgb}{0.82,0.91,0.84}

\begin{document}


\title{Storing and Querying Large-Scale Spatio-Temporal Graphs with
High-Throughput Edge Insertions}

\numberofauthors{1}
\author{
\alignauthor Mengsu Ding\hspace{1in}Muqiao
Yang$^\ddagger$\hspace{1in}Shimin Chen\titlenote{Shimin Chen is the
corresponding author.  Muqiao Yang contributed to the work as a summer
intern at ICT, CAS.  
}\\
\vspace{0.05in}
\begin{tabular}{c@{\hskip 0.3in}c}
\affaddr{SKL of Computer Architecture, ICT, CAS} &
\affaddr{$^\ddagger$Department of Electrical and Computer Engineering} \\
\affaddr{University of Chinese Academy of Sciences} &
\affaddr{Carnegie Mellon University} \\
\end{tabular}
}

\maketitle


\begin{abstract}

Real-world graphs often contain spatio-temporal information and evolve
over time.  Compared with static graphs, spatio-temporal graphs have
very different characteristics, presenting more significant challenges
in data volume, data velocity, and query processing.
In this paper, we describe three representative applications to
understand the features of spatio-temporal graphs.  Based on the
commonalities of the applications, we define a formal spatio-temporal
graph model, where a graph consists of location vertices, object
vertices, and event edges.  Then we discuss a set of design goals to
meet the requirements of the applications: (i) supporting up to 10
billion object vertices, 10 million location vertices, and 100
trillion edges in the graph, (ii) supporting up to 1 trillion new
edges that are streamed in daily, and (iii) minimizing cross-machine
communication for query processing.  
We propose and evaluate PAST, a framework for efficient
\underline{PA}rtitioning and query processing of
\underline{S}patio-\underline{T}emporal graphs.  Experimental results
show that PAST successfully achieves the above goals.  It improves
query performance by orders of magnitude compared with
state-of-the-art solutions, including JanusGraph, Greenplum, Spark and
ST-Hadoop.


\end{abstract}


\section{Introduction}
\label{sec:intro}


Graphs have been widely used to represent real-world entities and
relationships.  Real-world graphs often contain spatio-temporal
information generated by a wide range of hardware devices (e.g.,
sensors, POS machines, traffic cameras, barcode scanners) and software
systems (e.g., web servers).  We describe three representative
application scenarios in the following.

\Paragraph{Application 1: Customer Behavior Tracking and Mining}
Understanding customer behaviors is helpful for detecting fraud and
providing personalized services.  For example, credit card companies
track customers' credit card uses for fraud detection.  Internet
companies track users' browsing behaviors to achieve personalized
recommendations.  In a customer behavior tracking and mining
application, people (customers) and locations can be modeled as graph
vertices, while an edge linking a person vertex to a location vertex
represents the event that the person visits the location at a certain
time, as shown in Fig.\cref{fig:app1}.  This forms a spatio-temporal
graph.  People visiting similar locations at similar timestamps often
have similar personal interests.  In other words, it is desirable to
discover groups of people vertices that have similar edge structures
in the spatio-temporal graphs.

\begin{figure}[t]
  \centering
  \includegraphics[width=3in]{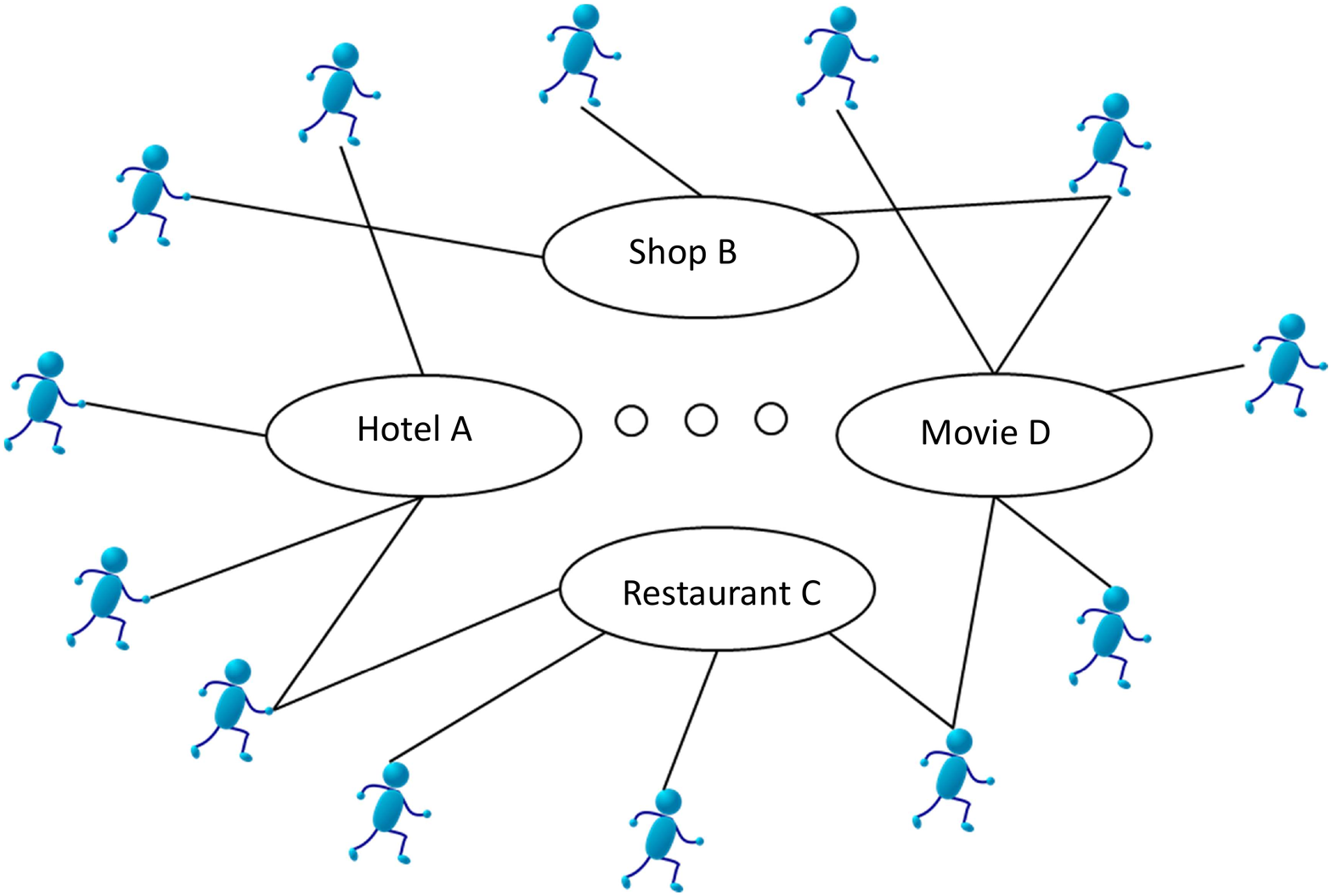}
  \vspace{-10pt}
  \caption{Customer behavior tracking and mining.}
  \label{fig:app1}
  \vspace{-0.2in}
\end{figure}

\Paragraph{Application 2: Clone-Plate Car Detection} A clone-plate car
displays a fake license plate that has the same license number as
another car of the same make and model.  In this way, the owner can
avoid annual registration and insurance fees, and/or purchase the car
without going through license plate lottery (which is a measure to
address the traffic jam problems in major cities in China).  However,
it is difficult to detect a clone plate since a query to the car
registration database will return a valid result.  
A promising approach is to exploit the large number of traffic cameras
on high ways and local roads to detect clone plates.  As a car passes
by a traffic camera, the camera takes a photo and automatically
recognizes the license number on the car plate. Car plates and cameras
can be modeled as vertices in a spatio-temporal graph. An edge
connecting a car plate vertex to a camera vertex indicates that the
camera records the car plate at a certain time.  Then, a clone plate
is detected if a car plate vertex has two edges whose locations are
far apart but timestamps are close such that it is impossible for the
car to cover the distance in such a short period of time.

\Paragraph{Application 3: Shipment Tracking} Recent years have seen
rapid growth of the shipping business.  E-commerce sites, such as
Amazon and Alibaba, have become increasingly popular.  Customers place
orders online and the ordered goods are delivered to their door steps
by shipping companies.  Each shipment package contains a barcode.
Shipping companies track packages by scanning the barcodes on packages
with barcode scanners in regional or local offices.  The problem of
shipment tracking can be modeled as a spatio-temporal graph.  Packages
and barcode scanners are represented as vertices.  The event that a
barcode scanner detects a package is recorded as an edge between the
package vertex and the scanner vertex.  In this way, the
spatio-temporal graph can be used to track shipment and answer
status-checking queries.

\Paragraph{Challenges for Storing and Querying Spatio-Temporal Graphs}
Based on the commonalities of the three applications, we define a
formal spatio-temporal graph model, where a graph consists of location
vertices (e.g., locations, cameras, barcode scanners), object vertices
(e.g., people, car plates, packages), and event edges that connect
them.  Compared with static graphs, spatio-temporal graphs pose more
significant challenges in data volume, data velocity, and query
processing:

\begin{list}{\labelitemi}{\setlength{\leftmargin}{5mm}\setlength{\itemindent}{0mm}\setlength{\topsep}{0.5mm}\setlength{\itemsep}{0mm}\setlength{\parsep}{0.5mm}}
\item \emph{Data Volume: $\sim$10 billion object vertices, $\sim$10
million location vertices, and $\sim$100 trillion edges.}
First, the requirement of $\sim$10 billion object vertices is based on
the fact that there are about 7.5 billion people in the world.
Second, according to booking.com, there are about 1.9 million hotels
and other accommodations in the world.  Suppose there are 5 times more
shops, restaurants, and theatres than hotels.  Then the total number
of locations can be on the order of 10 million.  Finally, suppose an
object vertex sees about 10 new edges per day.  If we store the edges
generated in the recent 3 years in the graph, there can be about 100
trillion edges in the graph.  The data volume of a spatio-temporal
graph can be much larger than static graphs.  


\item \emph{Data Velocity: up to 1 trillion new edges per day.}
Compared with static graphs, spatio-temporal graphs must support a
large number of new edges per day.  Suppose there are an average of 10
new edges per object vertex daily, and peak cases (e.g., Black Friday
shopping) can see ten times more activities.  This leads to a peak
velocity of 100 new edges per object vertex daily, i.e. 1 trillion new
edges in total per day.  


\item \emph{Query Processing: prohibitive communication cost.} The
data volume and velocity challenges entail a distributed solution that
stores the graph data on a number of machines.  However, without
careful designs, processing queries can easily lead to significant
cross-machine communication.  For example, in order to find subgroups
of people with similar interests (i.e. visiting similar locations at
similar time), it is necessary to combine the spatial information in
location vertices, the temporal information in edges, and properties
of person vertices in the computation.  As the data volume is huge,
the cross-machine communication cost can be prohibitively high.  

\end{list}

\noindent However, existing graph partitioning solutions~\cite{Metis,
Chaco, Scotch, DBLP:journals/pvldb/VermaLSG17} focus mostly on static
graphs, trying to minimize the number of cut edges across partitions
and balance the partition sizes.  Unfortunately, these solutions do
not take into account spatio-temporal characteristics, which are
important for query processing in spatio-temporal graphs.  On the
other hand, storing and querying spatio-temporal graphs using
state-of-the-art distributed graph database systems (e.g.,
JanusGraph~\cite{JanusGraph}), MPP relational database systems (e.g.,
Greenplum~\cite{Greenplum}), big data analytics systems (e.g.,
Spark~\cite{Spark}), or Hadoop enhanced for spatio-temporal data (i.e.
ST-Hadoop~\cite{DBLP:journals/pvldb/AlarabiM17}) result in poor
performance (cf. Section~\cref{sec:evaluate}).

\Paragraph{Our Solution: PAST}
In this paper, we propose and evaluate PAST, a framework for efficient
\underline{PA}rtitioning and query processing of
\underline{S}patio-\underline{T}emporal graphs.  We propose
diversified partitioning for location vertices, object vertices, and
edges.  We exploit the multiple replicas of edges to design
spatio-temporal partitions and key-temporal partitions.  Then we
devise a high-throughput edge ingestion algorithm and optimize the
processing of spatio-temporal graph queries.  Experimental results
show that PAST can successfully address the above challenges.  It
improves query performance by orders of magnitude compared to
state-of-the-art solutions, including JanusGraph, Greenplum, Spark,
and ST-Hadoop.

\Paragraph{Contributions} The contributions of our work are threefold:
\begin{list}{\labelitemi}{\setlength{\leftmargin}{5mm}\setlength{\itemindent}{0mm}\setlength{\topsep}{0.5mm}\setlength{\itemsep}{0mm}\setlength{\parsep}{0.5mm}}
  \item We define a formal model for spatio-temporal graphs, and
examine the design goals and challenges for storing and querying
spatio-temporal graphs.
  \item We propose PAST, a framework for efficient
\underline{PA}rtitioning and query processing of
\underline{S}patio-\underline{T}emporal graphs.  It consists of a
number of interesting features: (i) diversified partitioning for
different vertex types and different edge replicas; (ii) graph storage
with compression to reduce storage space consumption; (iii) a
high-throughput graph ingestion algorithm to meet the challenge of
data velocity; (iv) efficient query processing that leverages the
different graph partitions; and (v) a cost model to choose the best
graph partitions for query evaluation.
  \item We compare the efficiency of PAST with state-of-the-art
solutions, including JanusGraph, GreenPlum, Spark, and ST-Hadoop in
our experiments.  We design a benchmark based on the query workload of
the representative applications.  Experimental results show that PAST
achieves orders of magnitude better performance than state-of-the-art
solutions.
\end{list}

\Paragraph{Paper Organization} The remainder of the paper is organized
as follows. Section~\cref{sec:related} reviews related literature.
Section~\cref{sec:def} presents a formal definition of spatio-temporal
graphs.  Section~\cref{sec:sys} overviews the system architecture of
PAST.  Section~\cref{sec:core} elaborates the partitioning and storage
scheme for spatio-temporal graphs.  Section~\cref{sec:stream} describes
the high-throughput edge ingestion support in PAST.
Section~\cref{sec:optimize} presents optimizations on spatio-temporal
query processing.  Section~\cref{sec:evaluate} presents experimental
results.  Section~\cref{sec:discussion} discusses several interesting
issues.  Finally, Section \cref{sec:conclusion} concludes the paper.

\section{Related Work}
\label{sec:related}

We consider four solutions in the context of spatio-temporal graphs in
Section~\cref{subsec:existing}, and discuss more related work in
Section~\cref{subsec:other-related}.

\subsection{Supporting Spatio-Temporal Graphs with Existing Solutions}
\label{subsec:existing}

\Paragraph{Distributed Graph Database Systems} Graph database systems (e.g.,
JanusGraph~\cite{JanusGraph}, Titan~\cite{Titan}, Neo4j~\cite{Neo4j},
SQLGraph~\cite{DBLP:conf/sigmod/SunFSKHX15},
ZipG~\cite{DBLP:conf/sigmod/KhandelwalYY0S17}) often support the property graph
data model. We consider distributed graph database systems (e.g.,
JanusGraph~\cite{JanusGraph}, Titan~\cite{Titan}) for supporting
spatio-temporal graphs. They store vertices and edges in Key-Value stores
(e.g., Cassandra~\cite{Cassandra}, HBase~\cite{HBase}) and exploit search
platforms (e.g., Elasticsearch~\cite{Elasticsearch}, Solr~\cite{Solr}) as
indices for selective data accesses. Simple graph traversal queries can be
efficiently handled. However, they are inefficient for large-scale
spatio-temporal graphs because (i) they do not support direct filtering on time
or spatial ranges, which are frequently used in spatio-temporal queries, and
(ii) they need to scan the entire graph then invoke big data analytics systems
(e.g., Spark~\cite{Spark}) for complex queries, which incurs huge I/O overhead.

\Paragraph{MPP Relational Database Systems} We can store graph vertices and
edges as relational tables in MPP database systems (e.g.,
Greenplum~\cite{Greenplum}), and use SQL for querying spatio-temporal graphs.
Greenplum supports partitioning on multiple subsequent dimensions. The data is
first partitioned by the first partition dimension. Then the second partition
dimension is applied to each first-level partition to obtain a set of
second-level partitions, so on and so forth. The partitions in all dimensions
form a tree. MPP database systems can be inefficient for spatio-temporal
graphs because (i) there is no support for spatial partitions, and (ii) any
query has to start at the root-level and follow the tree even if the query does
not have filtering predicates on the first partition dimension, potentially
incurring significant CPU, disk I/O, and communication overhead.

\Paragraph{Big-data Analytics Systems} General-purpose big-data analytics
systems (e.g., Hadoop~\cite{Hadoop}, Spark~\cite{Spark}) support large-scale
computation on data stored in underlying storage systems, such as distributed
file systems (e.g., HDFS), distributed key-value stores (e.g.,
Cassandra~\cite{Cassandra}). We can store the spatio-temporal graphs in the
underlying storage systems, and run computation jobs on big-data analytics
systems for querying spatio-temporal graphs. However, the system needs to load
the entire graph before processing. It is not suitable for simple graph
traversal queries. Since there is no support for filtering on time or spatial
ranges on the underlying graph data, complex queries can see large unnecessary
overhead due to reading the entire graph.

\Paragraph{ST-Hadoop} ST-Hadoop~\cite{DBLP:journals/pvldb/AlarabiM17} is an
extension to Hadoop~\cite{Hadoop} and SpatialHadoop~\cite{EldawyM15}. It
represents previous studies that exploit multi-dimensional indices for
supporting spatio-temporal data~\cite{DBLP:conf/icde/PapadiasTKZ02,
DBLP:conf/sigmod/NgC04, DBLP:conf/vldb/RaseticSEN05, DBLP:conf/icde/AhmadN08,
DBLP:conf/icde/Lu0J11}, and provides a scalable solution when data volume is
large. ST-Hadoop organizes data into two-level indices. The first level is
based on the temporal dimension, while the second level builds spatial
indices. In this way, ST-Hadoop reduces data accessed for queries with
temporal and spatial range predicates, thereby achieving better performance
than Hadoop and SpatialHadoop. However, there are several disadvantages of
ST-Hadoop for supporting spatio-temporal graphs. First, ST-Hadoop sacrifices
storage for query performance. ST-Hadoop replicates its two-level indices into
multiple layers with different temporal granularities (e.g., day, month, year).
In each layer, the whole data set is replicated and partitioned. Second,
ST-Hadoop is inefficient in supporting streaming data ingestion. It needs to
sample across \emph{all data} to estimate the data distribution and compute
temporal and spatial boundaries. This basically requires temporarily storing
incoming data and then periodically shuffles the data to build the indices.
Depending on the temporal granularity, this may incur large temporal storage
overhead and huge bursts of computation. Third, there is no support for
indexing vertex IDs. Consequently, simple graph traversal queries may incur
significant I/O overhead for reading a large amount of data. Finally,
ST-Hadoop is based on the MapReduce framework, where the intermediate results
are written to disks, potentially incurring significant disk I/O overhead.

\subsection{Other Related Work}
\label{subsec:other-related}

There are a large number of static graph partitioning algorithms in the
literature, such as METIS~\cite{DBLP:journals/siamsc/KarypisK98},
Chaco~\cite{Chaco, DBLP:conf/sc/HendricksonL95},
Scotch~\cite{DBLP:conf/hpcn/PellegriniR96},
PMRSB~\cite{DBLP:conf/sc/Barnard95},
ParMetis~\cite{DBLP:journals/jpdc/KarypisK98},
Pt-Scotch~\cite{DBLP:journals/pc/ChevalierP08}. Several recent studies
introduce light-weight algorithms for partitioning large-scale dynamic
graphs~\cite{DBLP:journals/pvldb/HuangA16, DBLP:journals/pvldb/XuCC14,
DBLP:conf/kdd/StantonK12, DBLP:conf/sigmod/MondalD12}. However, none of these
proposed methods leverage spatio-temporal characteristics of the data, which
are important for query processing in spatio-temporal graphs.

Previous bipartite graph partitioning algorithms~\cite{DBLP:conf/kdd/Dhillon01,
DBLP:conf/kdd/GaoLZCM05} focus on simultaneous clustering with spectral
co-clustering by computing eigenvalues and eigenvectors of Laplacian matrices.
However, spatio-temporal graphs often have billions of vertices and trillions
of edges, causing huge matrix storage and calculation overhead.

Time series databases (e.g., TimescaleDB~\cite{Timescale},
LittleTable~\cite{DBLP:conf/sigmod/RheaWWAS17}) enhance relational databases
for supporting time-series data by partitioning rows by timestamp. In
addition, every partition can be further sorted / partitioned by a specified
key. However, there is no efficient support for spatial range predicates,
which are important for spatio-temporal graphs.

Spatio-temporal graphs have been employed for video
processing~\cite{DBLP:conf/sigmod/LeeOH05}. A video consists of a series of
frames. Here, a graph vertex corresponds to a segmented region in a video
frame. A spatial edge connects two adjacent regions in a frame, while a
temporal edge links two corresponding regions in two consecutive frames. In
this paper, we define a spatio-temporal graph model based on the representative
applications. It is a bipartite graph, which is quite different from the model
in the video processing context.

In our previous work, we proposed LogKV~\cite{DBLP:conf/cidr/CaoCLWW13}, a
high-throughput, scalable, and reliable event log management system. The
ingestion algorithm of PAST is an extension of the ingestion algorithm of
LogKV. However, LogKV focuses mainly on temporal events. There is no support
for spatial range predicates. Therefore, LogKV cannot efficiently process the
spatio-temporal query workload considered in this paper.
Moreover, a preliminary four-page version of this paper overviews the
high-level ideas and shows preliminary experimental
results~\cite{PAST-icde19-short}.

\section{Problem Formulation}
\label{sec:def}

In this section, we present the formal definition of spatio-temporal
graphs and the query workload, then examine the design challenges.

\subsection{Spatio-temporal Graph}

Based on the representative applications in Section \cref{sec:intro},
we define spatio-temporal graphs as follows:

\newtheorem{mydef}{\bfseries Definition}

\vspace{-0.10in} \begin{mydef}[Spatio-temporal Graph] \label{def:sg} 
A spatio-tem- poral graph $G = (V_L, V_O, E)$.  $V_L$ is a finite set
of location vertices.   Every location vertex contains a location
property.  $V_O$ is a finite set of object vertices that represent
objects being tracked.  Every vertex in $V_L$ and $V_O$ is assigned a
globally unique vertex ID.  $E$ is a set of undirected edges.  Every
edge in $E$ connects an object vertex to a location vertex, and
contains a time property.  
\end{mydef}

\vspace{-0.10in} Examples of location vertices include locations that
customers visit in customer behavior tracking and mining, traffic
cameras in clone-plate car detection, and barcode scanners in shipment
tracking.  Examples of object vertices include people in customer
behavior tracking and mining, car plates in clone-plate car detection,
and packages in shipment tracking.  Every location vertex contains a
location property such that given two location vertices $u$ and $v$,
their distance $dist(u,v)$ is well defined.  For example, if the
application is concerned about geographic locations, then the location
property consists of the latitude and longitude of the location
vertex. An edge contains object ID, timestamp, location ID, and other
application-dependent properties.

In essence, a spatio-temporal graph as defined in Definition 1 is an
undirected bipartite graph.  We do not consider edges between object
vertices and edges between location vertices.  This abstraction
captures the key characteristics and the main challenges of the three
representative applications.

\subsection{Query Workload} 
\label{subsec:query}

We consider the following four types of queries based on the
representative applications:
\begin{list}{\labelitemi}{\setlength{\leftmargin}{5mm}\setlength{\itemindent}{0mm}\setlength{\topsep}{0.5mm}\setlength{\itemsep}{0mm}\setlength{\parsep}{0.5mm}}
\item \textbf{Q1: object trace}. Given an object and a time range,
find the list of (object, timestamp, location)'s that represent the
locations visited by the object during the time range.  For example,
Q1 can display the trace of a shipment package or the activities of a
customer in a specified period of time.
\item \textbf{Q2: trace similarity}. Given two objects and a time
range, compute the similarity of the two object traces during the time
range.  Consider edge ($o1$, $t1$, $l1$) in object $o1$'s trace and
edge ($o2$, $t2$, $l2$) in object $o2$'s trace.  The two edges are
considered similar if $|t2-t1| \leq TH_{time}$ and $dist(l2, l1) \leq
TH_{dist}$, where $TH_{time}$ and $TH_{dist}$ are predefined
thresholds on time and location distance, respectively.  The
similarity of the two traces is the count of similar edge pairs in the
two traces.
\item \textbf{Q3: similar object discovery}. Given an object $o$ and a
time range, list the objects that have similar traces compared to $o$.
Display the list in the descending order of trace similarity with $o$.
Q3 can be used to discover people with similar interests in the
customer behavior tracking and mining application.
\item \textbf{Q4: clone object detection}. Given a time range,
discover all the clone objects.  An object $o$ is a clone object if
there exists two incident edges ($o$, $t1$, $l1$) and ($o$, $t2$,
$l2$) such that the computed velocity is beyond a predefined
threshold: $\tfrac{dist(l2, l1)}{|t2-t1|} > TH_{velocity}$.  Q4
supports clone-plate car detection.  It can also be used to detect
duplicate credit cards by a credit card company in the customer
behavior tracking and mining application.
\end{list}

\subsection{Understanding the Challenges}\label{subsec:goals}

\Paragraph{Data Volume} The goal is to support $\sim$10 billion object
vertices, $\sim$10 million location vertices, and $\sim$100 trillion
edges.  Suppose the properties of a vertex require at most 100B.  Then
the object vertices and the location vertices require about 1TB and
1GB space, respectively.  An edge contains at least (object ID,
timestamp, location ID).  Suppose each field takes 8B.  Then an edge
takes at least 24B. 100 trillion edges require 2.4PB space.  The
3-replica redundancy policy requires a total 7.2PB storage space.
Suppose the disk capacity of a machine is about 10TB.  Therefore, 100
trillion edges require on the order of 1000 machines to store.

\Paragraph{Data Velocity} The goal is to support up to 1 trillion new
edges per day.  This means 1Tx24B = 24TB/day of new ingestion data.  A
day consists of 86,400 seconds.  Thus, this requires the design to
support 290MB/s ingestion throughput.

\Paragraph{Query Processing} We would like to minimize cross-machine
communication in query processing.  A random partition scheme would
work poorly for Q1--Q4 because a large amount of unrelated data need
to be visited.   Since the four queries process time ranges, object
IDs, and spatial locations, it is desirable to organize the data for
efficient accesses using these dimensions.

\section{PAST Overview}
\label{sec:sys}

We propose PAST, a framework for efficient \underline{PA}rtitioning and
query processing of \underline{S}patio-\underline{T}emporal graphs.
The system architecture of PAST is shown in Fig.\cref{fig:arch}.
PAST consists of multiple machines connected through the data center
network in a data center.  We assume that the latencies and bandwidths
of the data center network are much better than those in wide area
networks.
There is a coordinator machine and a large number of (e.g., 1000)
worker machines.  The coordinator keeps track of meta information of
the graph partitions and coordinates data ingestion.  The workers
store the graph data, handle incoming updates, and process queries.


\begin{figure}[t]
\vspace{-0.1in}
\centerline{\includegraphics[width=3.3in]{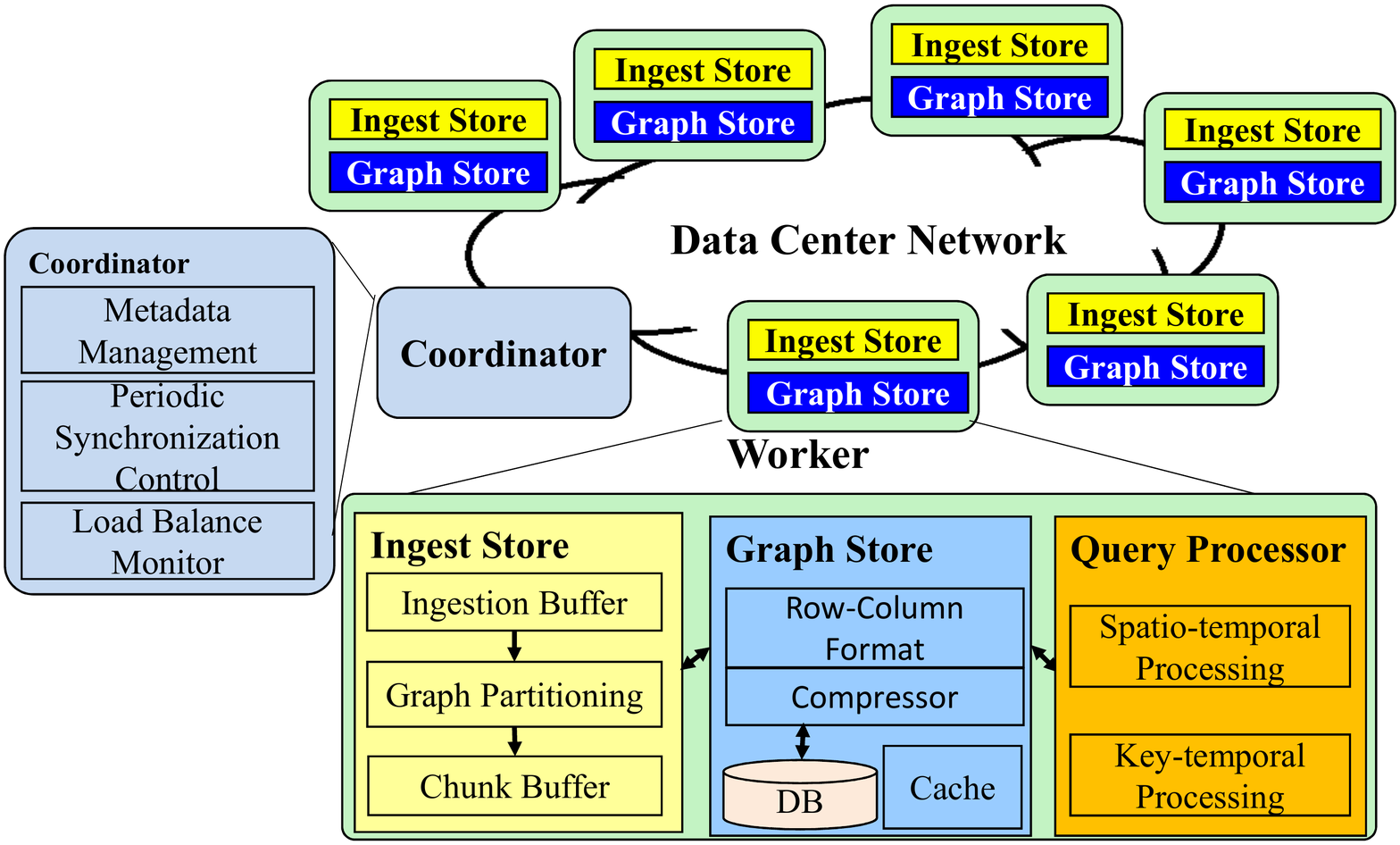}}
\vspace{-0.30in}
\caption{PAST system architecture.}
\label{fig:arch}
\vspace{-0.25in}
\end{figure}

\Paragraph{Graph Partitioning and Storage} The GraphStore component in
Fig.\cref{fig:arch} implements PAST's graph partitioning methods and
supports compressed storage of the main graph data.  The GraphStore
utilizes an underlying DB/storage system (e.g., the Cassandra
key-value store in our implementation).

We propose diversified partitioning for location vertices and
object vertices because they have drastically different
characteristics.  The number of location vertices is 1/1000 of that of
object vertices.  As the spatio-temporal graph is a bipartite graph,
the average degree of a location vertex is 1000 times of that of an
object vertex.  As a result, they have very different impact on the
communication patterns in query processing. 

Edges consume much more space than vertices, and are often the
performance critical factor in query processing.  All four queries
filter the edges with a given time range.  Therefore, GraphStore
should support time range filters efficiently.  Q1, Q2, and Q4 access
edges for a given object, two objects, and all objects, respectively.
Hence, it would be nice to organize edges according to object IDs.  On
the other hand, Q3 can be more efficiently computed if edges are
stored in spatio-temporal-aware orders so that GraphStore can filter
out a large number of edges that are not relevant to the trace of the
specified object.  However, these requirements seem contradicting.  
We solve this problem by taking advantage of the multiple replicas of
edges.  For the purpose of fault tolerance, PAST stores multiple
replicas for edges (e.g., 3 replicas).  Therefore, we propose a
spatio-temporal edge partition method and a key-temporal edge
partition method for different edge replicas.

\Paragraph{High-throughput Streaming Edge Ingestion} New edge updates
are streamed in rapidly.  As shown in Fig.\cref{fig:arch}, the
IngestStore component maintains a staging buffer for incoming new
edges.  All the worker machines handle incoming edges in rounds.  They
keep loosely synchronized clocks.  In every round, the IngestStores
collect the incoming edges in the current round.  At the same time,
IngestStores send edges collected in the previous round to their
destination GraphStores based on the partitions computed by PAST's
partition methods.  We design an efficient algorithm to perform the
data ingestion that avoids hot spots in the data shuffling.

\Paragraph{Query Processing and Optimization} Given PAST's partition
methods, we design a cost-based query optimizer to choose the best
partitions for an input query.  Our goal is to reduce cross-machine
communication and edge data access as much as possible.  The edge
partitions divide the spatio-temporal space and key-temporal space
into discretized blocks.  We perform block-level filtering to avoid
reading irrelevant edge data.  Then we also take advantage of triangle
inequalities for finer-grain filtering if geographic locations are
used in the application.

\newtheorem{mytheorem}{\bfseries Theorem}
\newtheorem{myproof}{\bfseries Proof}
\newtheorem{mylemma}{\bfseries Lemma}

\section{\hspace{-1mm}Graph Partitioning and Storage}
\label{sec:core}

In this section, we describe the partitioning and storage methods
for spatio-temporal graphs in PAST.

\subsection{Diversified Partitioning}
\label{subsec:partitioning}

There are two main approaches to graph partitioning in the literature:
vertex-based partitioning and edge-based partitioning. In vertex-based
partitioning~\cite{DBLP:journals/siamsc/KarypisK98}, a vertex is the
basic partitioning unit.  It assigns vertices along with their
incident edges to partitions in order to minimize cross-partition
edges.
However, for a high-degree vertex, which is common in real-world
graphs, there will be a large number of cross-partition edges no
matter which partition to assign it.  To address this problem,
edge-based partitioning~\cite{DBLP:conf/osdi/GonzalezLGBG12} assigns
edges to partitions.  If the incident edges of a (high-degree) vertex
are in $k>1$ partitions, then the scheme chooses one partition to
store the main copy of the vertex, and creates a ghost copy of the
vertex in every other partition, thereby reducing the number of
cross-partition edges to $k-1$ ghost-to-main virtual edges.

We propose diversified partitioning for spatio-temporal graphs.  Our
solution is inspired by edge-based partitioning.  It considers the
different properties of location vertices, object vertices, and edges,
and the characteristics of spatio-temporal graph queries.

\Paragraph{Location Vertex} The $\sim$10 million location vertices
require about 1GB space (cf. Section~\cref{subsec:goals}). A mid-range
server machine today is often equipped with 100GB--1TB of main memory.
Therefore, all the location vertices can easily fit into the main
memory of one machine.  On the other hand, location vertices are
frequently visited for obtaining the locations of specified location
IDs.  Therefore, PAST stores all location vertices in every worker
machine, and loads them into main memory at system initialization
time. In this way, location information can be efficiently accessed
locally without cross-machine communication.  PAST updates all the
worker machines when a location is updated.  The update cost is
insignificant because location vertices (e.g., shops, hotels, traffic
camera locations, shipping services) change very slowly, 

\Paragraph{Object Vertex} PAST performs hash-based partitioning for
object vertices.  The scheme is inspired by Redis~\cite{Redis}. Each
object vertex is assigned a unique 8-byte vertex ID. We first divide
the vertices into slots. There are $2^p$ slots. Given a vertex ID
$vid$, PAST computes $s=hash(vid)$, where $s \in[0,2^p)$ is a $p$-bit
slot ID. Then, we assign the slots to the worker nodes in a round
robin manner: $w = i \bmod{N}$, where $N$ is the number of worker
nodes. For fault tolerance purpose, we create 3 replicas by storing
the slot to worker $w$, $w+1$, and $w+2$. 

\Paragraph{Edge} An edge contains a triplet (object ID, timestamp,
location ID). Generally speaking, queries in spatio-temporal
applications often contain filtering predicates on object IDs, time
ranges, and locations. Specifically, Q1-Q4 all have time range
filters. Q1, Q2, and Q4 will benefit from a data layout where edges of
each object are stored together, while the amount of data accessed by
Q3 is reduced if spatio-temporal filtering is efficiently supported.
We take advantage of the multiple edge data replicas to design a
twofold partitioning strategy as described in Section~\cref{subsec:sp}
and~\cref{subsec:kp}.

\subsubsection{\hspace{-2mm}Skew-aware Spatio-temporal Edge Partitioning}
\label{subsec:sp}

Spatio-temporal edge partitioning first partitions edges according to
the spatial dimension so that given an edge $(o,t,l)$, all edges
associated with locations near $l$ are likely to be in the same
partition.  Then within a spatial partition, it constructs temporal
sub-partitions, each of which contains edges in a disjoint time range.
For the second part, it is straightforward to sort the edges by time
and obtain the temporal sub-partitions.  Therefore, we focus on the
spatial partitioning part of the design in the following.

\begin{algorithm}[t]
\caption{Unbounded weight-based spatial mapping.}
\label{algo:sp-unbounded}

   \SetKwProg{Fn}{Function}{}{}

   \Fn{UnboundedSpatioMapping(number of machines $M$,
         region list $\mathcal{R}$, location vertex list $V_L$)}{

        $\mathcal{R}_z$= Encode and sort regions in $\mathcal{R}$ by Z-codes;

        \lForEach{$r \in \mathcal{R}_z$}{Initialize weight $\mathcal{W}_r=0$}

        \lForEach{$l \in V_L$}{
            $r$ = $LocToRegion$($l$);
            $\mathcal{W}_{r}$++
        }

        $sum$=$0$; $rset$ = $\emptyset$;

        \ForEach{$r \in \mathcal{R}_z$}{

            $sum$+=$\mathcal{W}_{r}$;
            $rset$ = $rset \cup \{r\}$;

            \If{$sum >= |V_L|/M$}{
                Assign $rset$ to the next machine;

                $sum$=$0$; $rset$ = $\emptyset$;
            }
        }

        \lIf{$rset$ != $\emptyset$}{Assign $rset$ to the next machine}
    }

\end{algorithm}
\begin{algorithm}[t]

\caption{Bounded weight-based spatial mapping.}
\label{algo:sp-bounded}

   \SetKwProg{Fn}{Function}{}{}

   \Fn{BoundedSpatioMapping(
          number of machines $M$,
          region list $\mathcal{R}$,
          location vertex list $V_L$,
          unit width $b$
        )}{

        \lForEach{$r \in \mathcal{R}$}{Region weight 
$\mathcal{W}_r=0$}  \label{algo_line:init_weight_start}

        \lForEach{$l \in V_L$}{
            $r$ = $LocToRegion$($l$);
            $\mathcal{W}_{r}$++
        } \label{algo_line:init_weight_end}


	Suppose regions in $\mathcal{R}$ form a $h_x \times h_y$ 2D
matrix. Group every $b\times b$ sub-matrix into a unit. The
list of $\lceil \tfrac{h_x}{b} \rceil \times \lceil \tfrac{h_y}{b}
\rceil$ units is denoted $\mathcal{U}$;   \label{algo_line:region_units}

	\lForEach{$u \in \mathcal{U}$}{$\mathcal{UW}_u = \sum_{r \in u}
\mathcal{W}_r$}  \label{algo_line:unit_weights}

        $\mathcal{U}_{sorted}$ = sort $\mathcal{U}$ in descending
$\mathcal{UW}_u$ order;   \label{algo_line:sort_unit}

        \lForEach{$i$}{Machine weight $sum_i=0$}     \label{algo_line:assign_start}

	\ForEach{$u \in \mathcal{U}_{sorted}$}{

            $i_{min}$ = ${\operatorname{argmin}} (sum_i)$;

            Assign $u$ to machine $i_{min}$;
            
            $sum_{i_{min}}$ += $\mathcal{UW}_u$;   \label{algo_line:assign_end}
        }
    }

\end{algorithm}

We divide the universe of locations (e.g. the national map) into
spatial partitions.
We apply a grid to the universe.  Each grid cell is a partition (a.k.a region).
For simplicity, we consider square cells in this paper.  The cell width is denoted as $a$.
We choose $a$ such that (i) $a \gg D$, where $D=TH_{dist}$ is the distance threshold in Q2 and Q3, (ii) regions are small compared to the universe, and (iii) the number of regions are large compared to the number of worker machines.
Constraint (i) ensures that distance similarity computation can be handled inside individual partitions in most cases.  Constraint (ii) and (iii) make the partition size small enough so that multiple partitions can be assigned to a worker machine to achieve more balanced load.

Let partition $r$'s weight $\mathcal{W}_r$ be the number of locations in
region $r$.   The distribution of $\mathcal{W}_r$ can be very skewed.
For example, there are usually more locations in regions with higher
population.  More object vertices (e.g., people) may visit such
regions, leading to higher number of event edges.  Therefore, we
assume that $\mathcal{W}_r$ is proportional to the number of edges
incident to location vertices in region $r$.  

We would like to assign partitions to worker machines to achieve three goals:
(i) ensure that $Weight_i = \sum_{r \in Worker_i} \mathcal{W}_r$ of
worker machines are similar;
(ii) ensure that adjacent partitions are on the same machine with large probability
for reducing communication cost of evaluating spatial predicates; and
(iii) enable multiple workers to evaluate a query in parallel for better performance.  

In what follows, we propose an \emph{UnboundedSpatioMapping} algorithm
(Alg.\ref{algo:sp-unbounded}) that achieves the first two goals but
fails for the third goal. 
Then we design a \emph{BoundedSpatioMapping} algorithm (Alg.\ref{algo:sp-bounded}), 
and compute the algorithm parameter $b$ for achieving the third goal.

Alg.\ref{algo:sp-unbounded} lists the \emph{UnboundedSpatioMapping}
algorithm.  To achieve goal (ii), it employs the Z-order curve.  
It sorts the regions by their Z-codes (Line 2), and assigns regions in
contiguous Z-code intervals to machines (Line 5--10).  In this way,
neighboring regions are assigned to the same machine with large
probability.  
Theorem~\ref{thm:unbounded} shows that the algorithm achieves goal (i).

\begin{mytheorem}
\label{thm:unbounded}
Suppose that $\mathcal{W}_r\leq \epsilon_1 \tfrac{|V_L|}{M}$, where
$\epsilon_1 \ll 1$ and $\tfrac{|V_L|}{M}$ is the average load per
worker machine.  Given a target load skew $\epsilon_2$, then
$\tfrac{sum_i}{|V_L|/M}-1 < \epsilon_2$ (in
Alg.\ref{algo:sp-unbounded}).
\end{mytheorem}

\begin{myproof}
$\mathcal{W}_r\leq \epsilon_1 \tfrac{|V_L|}{M}$ means that the weight of an individual region is much smaller than the average machine weight.  This is guaranteed by constraints (ii) and (iii) for choosing $a$.
Line 8 in Alg.\ref{algo:sp-unbounded} ensures that the weight assigned to a worker 
machine $sum-W_r < \tfrac{|V_L|}{M}$.
Thus, we have $sum_i < (1+\epsilon_2)\tfrac{|V_L|}{M}$.
\end{myproof}

However, Alg.\ref{algo:sp-unbounded} fails to achieve goal (iii).  
Consider the trace of an object $o$.  $o$ may travel to a few cities. 
Those cities can be quite far away from each other, and visited
locations may lie in an area $A$ (e.g., California), which contains
a large number of partitions.
However, Alg.\ref{algo:sp-unbounded} may assign the entire $A$ to a single machine 
(e.g., our experiments use
10 worker machines).  As a result, the partitions do not allow
multiple workers to process the task in parallel.

Alg.\ref{algo:sp-bounded} solves the problem by assigning $b\times b$ adjacent spatial 
partitions to a worker at a time, where $b\geq 1$ is a control parameter.
The solution is based on the observation that 
the locations inside a city are often in spatial partitions that are close to each other.
With a larger $b$, it is more likely that locations in the same city of a trace are stored on the same machine, and therefore the amount of cross-machine communication may be reduced.
However, when $b$ is too large, the number of workers that can process a single query are roughly equal to the number of cities that $o$ visits, which significantly limits the parallelism of query processing.
Therefore, Alg.\ref{algo:sp-bounded} has to choose an appropriate $b$ to strike a
balance between communication cost and parallelism.

Alg.\ref{algo:sp-bounded} works as follows. It computes the weight $\mathcal{W}_r$ of every region $r$ (Line~\ref{algo_line:init_weight_start}--\ref{algo_line:init_weight_end}).
Then, it groups $b \times b$ adjacent regions into units (Line~\ref{algo_line:region_units}), 
and computes the weight $\mathcal{UW}_u$ of every unit, which is the sum of the weights of all the regions in this unit (Line~\ref{algo_line:unit_weights}).
The algorithm assigns units to worker machines and balances the loads.
To do this, it sorts the units in the descending order of unit weights (Line~\ref{algo_line:sort_unit}),
then employs a greedy algorithm to always assign the next heaviest unit to the machine with the smallest load (Line~\ref{algo_line:assign_start}--\ref{algo_line:assign_end}).

We derive constraints on parameter $b$ to satisfy the above goals. 
Theorem~\ref{thm:b1} shows the condition to balance the load among all worker machines, and
Theorem~\ref{thm:b2} shows the condition to balance communication  cost  and  parallelism.

Combining Theorem~\ref{thm:b1} and Theorem~\ref{thm:b2} and the conditions for $b$, we have 
$\tfrac{D}{2a(1-\sqrt{\alpha'})} \leq b \leq \sqrt{\epsilon_2 / \epsilon_1}$, where $a \gg D$, $D=TH_{dist}$, $\epsilon_1 \ll 1$ 
(e.g., $\epsilon$=0.001) and $\alpha'=max(\alpha,\tfrac{(2a-D)^2}{4a^2})$.

\begin{mytheorem}
\label{thm:b1}
Suppose that $\mathcal{W}_r\leq \epsilon_1 \tfrac{|V_L|}{M}$, where $\epsilon_1 \ll 1$ and $\tfrac{|V_L|}{M}$ is the average load per worker machine.  Given a target load skew $\epsilon_2$, if  $b \leq \sqrt{\epsilon_2 / \epsilon_1}$, then $\tfrac{sum_i}{|V_L|/M}-1 < \epsilon_2$.
\end{mytheorem}

\begin{myproof}
$\mathcal{W}_r\leq \epsilon_1 \tfrac{|V_L|}{M}$ means that the weight of an individual region is much smaller than the average machine weight.  This is guaranteed by constraints (ii) and (iii) for choosing $a$.
The greedy algorithm ensures that the weight assigned to a worker machine $sum_i < \tfrac{|V_L|}{M} + max(\mathcal{UW}_u)$. 
As each unit consists of $b^2$ regions, $max(\mathcal{UW}_u) \leq b^2 max(\mathcal{W}_r) \leq b^2 \epsilon_1 \tfrac{|V_L|}{M}$.
Combining the two inequalities and the condition for $b$, we have $sum_i < (1+\epsilon_2)\tfrac{|V_L|}{M}$.
\end{myproof}

\begin{mytheorem}
\label{thm:b2}
Given a target probability $\alpha$, 
if $b \geq \tfrac{D}{2a(1-\sqrt{\alpha'})}$, where $\alpha'=max(\alpha,\tfrac{(2a-D)^2}{4a^2})$, 
then $\forall$ unit $\mathcal{U}$, location $l$ and $l_{nr}$: 
$ Prob[l_{nr}\in \mathcal{U} |  l  \in \mathcal{U} \wedge dist(l_{nr}, l)\leq D] \geq \alpha$.
\end{mytheorem}

\begin{proof}
We divide a $ba \times ba$ unit (each unit is $b \times b$ adjacent
spatial partitions, each partition is a cell with width $a$) into 9 parts, 
as illustrated in Fig.\ref{fig:unit_divide}a.  $S_1$ is an inner square.  $P[l_{nr}
\notin U | l \in S_1]=0$.  $S_2$ is a $(ba-2D) \times D$ rectangle.
There is a single unit adjacent to $S_2$.  We consider the four $S_2$
similarly.
$P[l \in \cup S_2] =\tfrac{4(ba-2D)D}{(ba)^{2}}$.  
$S_3$ is a $D \times D$ square.  There are three units (including the
one in the diagonal position) adjacent to $S_3$.  
We also consider the four $S_3$ similarly. 
$P[l \in \cup S_3] =\tfrac{4D^2}{(ba)^{2}}$.

\begin{figure}[h]
  \vspace{-0.15in}
  \centering
  \includegraphics[width=3.3in]{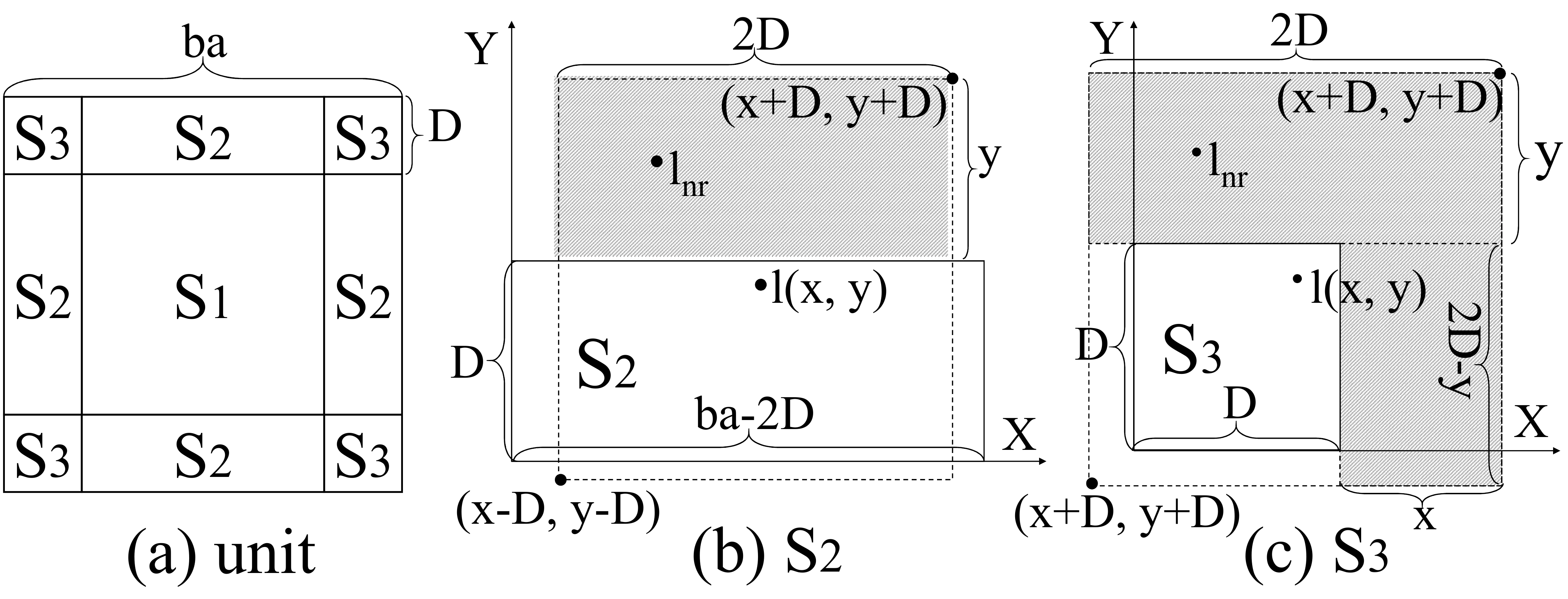}
  \vspace{-0.25in}
  \caption{Dividing a unit into nine parts.}
  \label{fig:unit_divide}
\vspace{-0.15in}
\end{figure}

We derive bounds on the probability by using Chebyshev distance
$dist_c(l_1,l_2)=max(|x_1-x_2|, |y_1-y_2|)$.  Note that Euclidean
distance $dist(l_1,l_2)\geq dist_c(l_1,l_2)$.  Given a location in
$S_2$ ($S_3$), Fig.\ref{fig:unit_divide}b (Fig.\ref{fig:unit_divide}c)
illustrates the shadow area outside of the unit where its nearby
location may reside in.  We have

\vspace{0.05in}
$\begin{array}{l}

 \hspace{-0.2in}

 Prob[l_{nr}\notin \mathcal{U} | l \in \mathcal{U} \wedge dist(l_{nr}, l)\leq D] 

 \vspace{0.05in}\\\hspace{-0.2in}

 = Prob[l_{nr} \notin \mathcal{U} | l \in S_2]Prob[l \in \cup S_2] 

 \vspace{0.05in}\\\hspace{-0.2in}

   + Prob[l_{nr} \notin \mathcal{U} | l \in S_3]Prob[l \in \cup S_3] \\

 \vspace{-0.05in}\\\hspace{-0.2in}

 \leq

\frac{\iint_{S_{2}}\frac{2Dy}{4D^{2}}dxdy}{(ba-2D)D}\frac{4(ba-2D)D}{(ba)^{2}}
+
\frac{\iint_{S_{3}}\frac{2Dx+2Dy-xy}{4D^{2}}dxdy}{D^{2}}\frac{4D^{2}}{(ba)^{2}}  

 \vspace{0.08in}\\\hspace{-0.2in}

 = \frac{4baD-D^{2}}{4b^{2}a^{2}}

\vspace{0.05in}

\end{array}$

$\begin{array}{l}

 \hspace{-0.2in}

 Prob[l_{nr}\in \mathcal{U} | l \in \mathcal{U} \wedge dist(l_{nr}, l)\leq D] 

 \vspace{0.05in}\\\hspace{-0.2in}

 = 1 -  Prob[l_{nr}\notin \mathcal{U} | l \in \mathcal{U} \wedge dist(l_{nr}, l)\leq D]

 \vspace{-0.05in}\\\hspace{-0.2in}

 = 1- \frac{4baD-D^{2}}{4b^{2}a^{2}}

\vspace{0.05in}

\end{array}$

We see that if
$1-\tfrac{4baD-D^{2}}{4b^{2}a^{2}} \ge \alpha$, then the target is achieved.

The valid solution to this inequality is $b\geq
\frac{D}{2a(1-\sqrt{\alpha})}$, if $\alpha \in [\frac{(2a-D)^2}{4a^2}, 1)$. 
Thus, we have 
if $b \geq \tfrac{D}{2a(1-\sqrt{\alpha'})}$, where $\alpha'=max(\alpha,\tfrac{(2a-D)^2}{4a^2})$, 
then $\forall$ unit $\mathcal{U}$, location $l$ and $l_{nr}$: 
$ Prob[l_{nr}\in \mathcal{U} |  l  \in \mathcal{U} \wedge dist(l_{nr}, l)\leq D] \geq \alpha$.

\end{proof}


\subsubsection{Key-temporal Edge Partitioning}
\label{subsec:kp}

Key-temporal edge partitioning first constructs key partitions to
store edges of objects in the same slot.  It follows the partitioning
method of object vertices to obtain the key partitions. That is, given
an edge (object ID, timestamp, location ID), it uses object ID to
compute the slot ID $s$ and then maps $s$ to worker $w$.  Then within
a key partition, it constructs temporal sub-partitions.  Similar to
spatio-temporal partitioning, it divides edges into disjoint time
ranges, each of which constitutes a sub-partition.

\subsection{Compressed Columnar Edge Store}
\label{subsec:store}

In this section, we focus on the storage of edges in a worker machine
node.  GraphStore exploits existing storage / DB systems (e.g.,
Cassandra in our implementation) as the underlying DB to store edge
data.

GraphStore organizes edge data into (row key, edge data) pairs, where
the row key uniquely identifies a spatio-temporal / key-temporal
sub-partition and the edge data contain compressed list of (object ID,
timestamp, location ID, other edge properties) in the sub-partition.
The row key is a concatenation of the following fields:
\begin{list}{\labelitemi}{\setlength{\leftmargin}{5mm}\setlength{\itemindent}{0mm}\setlength{\topsep}{0.5mm}\setlength{\itemsep}{0mm}\setlength{\parsep}{0.5mm}}
\item \emph{node id:} uniquely identifies the worker node\footnote{
\small Our implementation modifies the partitioning function of
Cassandra to identify the node id field in row keys for computing
Cassandra partitions.};

\item \emph{partitioning method:} `A' for spatio-temporal partitioning
and `B' for key-temporal partitioning;

\item \emph{partition-specific id:} region ID in spatio-temporal
partitioning or slot ID in key-temporal partitioning;

\item \emph{time range:} TimeRange = $\lfloor \mbox{timestamp/TRU} \rfloor$, where TRU (Time
Range Unit) is a configuration parameter used to discretize time.

\end{list}

\noindent For the edge data, we employ columnar layout for the
attributes then compress the columns.  Columnar layout is attractive
because (i) it has good compression ratio and (ii) a query needs to
uncompress and access only the relevant columns.  For example, we put
the object IDs of all edges in the sub-partition in an array, then
compress the array.  Similarly, we obtain the compressed
representations of timestamps, location IDs, and other edge properties
if exist.  We measure the compression ratios and efficiency of several
well-known compression algorithms, and choose LZ4 and Snappy in PAST
because of their good performance~\cite{Compressor}.

To store the columns in Cassandra, we have two implementations: (i)
\emph{C}: we create multiple tables (or column families) in Cassandra,
and store the compressed column of each edge attribute in a separate
table; and (ii) \emph{R}: we concatenates all the compressed columns
into a single binary value, and store the value in a single Cassandra
table.  \emph{R} essentially implements the PAX
layout~\cite{DBLP:conf/vldb/AilamakiDHS01}.  We evaluate the two
implementations experimentally.

\subsection{Fault Tolerance}
\label{subsec:fault-tolerance}

PAST maintains edge replicas using different partitioning methods.
However, this might cause two replicas of the same edge to reside in
the same physical machine.

In our design, we check if such situation occurs and store the edge to
another machine to ensure replicas are on different
machines.  Suppose an edge is assigned to worker $w_s$ in
spatio-temporal partitioning, and to $w_k$ in key-temporal
partitioning. The problem occurs if $w_s$=$w_k$. If the condition
is true, then PAST will store a copy of the edge to worker $w_s+1$.

We compute the extra space required.  Suppose that there are $M$
worker nodes. Edges are evenly distributed across the workers. Then
the probability that an edge is assigned to a worker is
$\tfrac{1}{M}$. The probability that $w_s$=$w_k$
is $M\cdot\tfrac{1}{M}\cdot\tfrac{1}{M} = \tfrac{1}{M}$. Therefore,
the additional storage incurred is $\tfrac{1}{M}$ of the total edge
data size. When $M$ is large (e.g., 100--1000), the extra space
required is negligible.

\section{\hspace{-3mm}High-Throughput Edge Ingestion}
\label{sec:stream}

The data velocity challenge requires PAST to support up to 1 trillion
new edges per day.  We break down this goal into three sub-goals: (i)
store up to 1 trillion new edges per day; (ii) achieve the proposed
partitioning and storage strategy for the new edges as described in
Section~\cref{sec:core}; and (iii) balance the ingestion workload
across the worker machines and avoid hot spots as much as possible.
While directly reflecting the desired ingestion throughput, sub-goal
(i) is not sufficient by itself.  The other two sub-goals are
important because sub-goal (ii) supports sustained ingestion
performance and enables query processing on new edges, and sub-goal
(iii) improves the scalability of the system.

We propose a high-throughput edge ingestion algorithm as described in
Alg.\cref{algo:random_shuffle}, which extends our previous work on
event log processing~\cite{DBLP:conf/cidr/CaoCLWW13} to support
diversified partitioning for spatio-temporal graphs.
The high-level picture of the algorithm is as follows.  IngestStore on
each machine buffers the incoming new edges before shuffling them to
GraphStore to implement the desired partitioning strategy.
IngestStore employs double buffering.  It buffers $m$-TRU worth of
data while shuffling the previous $m$-TRU worth of data to GraphStore
on the destination workers at the same time. $m$ is a parameter to
tolerate instant ingestion bursts at individual IngestStores and avoid
hot spots in shuffling.

Alg.\cref{algo:random_shuffle} consists of three functions.  The first
function, \emph{IngestStore\_AppendNewEdge}, is invoked by IngestStore
upon receiving a new edge.  It appends the new edge to the end of
\emph{inbuf} (Line~\cref{algo:append-2-inbuf}).  

Every $m$-TRU time, the coordinator initiates a new round of shuffling
operation by broadcasting a \emph{NextRound} message with a
$t_{start}$ parameter to all workers.  Then IngestStore at each worker
invokes the second function, \emph{IngestStore\_ComputePartition}, to
compute the partitions for edges in $[t_{start},
t_{start}+m\cdot\mbox{TRU})$ and copy the edges to \emph{outbuf}.
Note that the coordinator can initiate the round at time
$t_{start}+m\cdot\mbox{TRU}+t_{delay}$ in order to tolerate
communication delays from event sources for up to $t_{delay}$ time.
The function computes the spatio-temporal partition (Line 7--9) and
the key-temporal partition (Line 10--12) for an edge.  Then it checks
if the destination workers of the spatio-temporal and the key-temporal
partitions collide.  In such cases, it copies the edge to the
\emph{outbuf} of the next worker as discussed in
Section~\cref{subsec:fault-tolerance} (Line 15--16).  In the end, the
function truncates the \emph{inbuf}.  After the invocation,
IngestStore replies a \emph{PartitionDone} message to the coordinator.
Note that the copy operation is very similar to the partitioning step
in the in-memory partitioned hash join algorithm.  When the number of
destination workers is large, there can be significant TLB and cache
misses.  We perform multi-pass copying in the spirit of the
radix-cluster algorithm~\cite{DBLP:conf/vldb/ManegoldBK00} for better
CPU cache performance.  

When all the partitions are computed at all workers, the coordinator
broadcasts a \emph{Shuffle} message to all workers.  Then GraphStore
at each worker invokes the third function, \emph{GraphStore\_Shuffle}.
It randomly permutes the worker list, then attempts to retrieve edges
from IngestStore at every worker $w$ in the list
$\mathcal{W}_{random}$ (Line 21--23).  If $w$ is busy serving another
GraphStore, then it puts $w$ into the busy list (Line 22).  After
processing $\mathcal{W}_{random}$, the function repeatedly processes
workers in the busy list until the busy list is empty (Line 24--29).
Upon receiving edge data from all workers, GraphStore sorts and
compresses the edges, then stores the compressed data as described in
Section~\cref{subsec:store} (Line 30).

We re-examine the three sub-goals.  It is clear that sub-goal (ii) is
achieved by \emph{Ingest\_StoreComputePartition} and
\emph{GraphStore\_Shuffle}.  For sub-goal (iii), the double buffering
mechanism and the random permutation are designed to reduce hot spots
as much as possible.  For sub-goal (i), we measure the ingestion
throughput experimentally in Section~\cref{sec:evaluate}.

\begin{algorithm}[t]
  \caption{High-throughput edge ingestion algorithm.}
  \label{algo:random_shuffle}

  \SetKwProg{Fn}{Function}{}{}

  \Fn{IngestStore\_AppendNewEdge(objId,time,locId)}{
      Append edge (objId,time,locId) to inbuf;   \label{algo:append-2-inbuf}
  }

  \Fn{IngestStore\_ComputePartition($t_{start}$, $m$, 
     Worker list $\mathcal{W}=\{w_0,\cdots,w_{n-1}\}$)}{

     \ForEach{(objId,time,locId) $\in$ inbuf \textbf{and} \\ 
         \hspace{3em} time $\in [t_{start}, t_{start}+m\cdot\mbox{TRU})$}{

         tr= time / TRU;

         \vspace{1mm}
         \hspace{5mm} /* Compute spatio-temporal partition */

         $w_s$= $ComputeWorkerforLocation$(locId);

         Append (objId,time,locId) to outbuf$_s$[$w_s$][tr];

         \vspace{1mm}
         \hspace{5mm} /* Compute key-temporal partition */

         $w_k$= $ComputerWorkerforObj$(objId);

         Append (objId,time,locId) to outbuf$_k$[$w_k$][tr];

         \vspace{1mm}
         \hspace{5mm} /* check if need extra replica for fault tolerance */

         \If{$w_s$ == $w_k$}{ 

             $w_{next}$ = $GetNextWorker$($\mathcal{W}$, $w_s$);

             Append (objId,time,locId) to outbuf$_f$[$w_{next}$][tr];
         }
     }

     Truncate inbuf to free the space of the processed edges;
  }

  \Fn{GraphStore\_Shuffle($\mathcal{W}=\{w_0,\cdots,w_{n-1}\}$)}{
     $me$ = My worker Id;

     $\mathcal{W}_{random}$ = Randomly permute the worker list $\mathcal{W}$;

     \ForEach{$w \in \mathcal{W}_{random}$}{
         \lIf{Worker $w$ is busy}{
             Put $w$ into $busylist$}
         \lElse{
	     Retrieve outbuf$_s$[$me$][..], outbuf$_k$[$me$][..], and
             outbuf$_f$[$me$][..] from worker $w$;
             }
     }

     \Repeat{ $busylist$ == $\emptyset$ }{
         \ForEach{$w \in busylist$}{
             \If{Worker $w$ is not busy}{
	         Retrieve outbuf$_s$[$me$][..], outbuf$_k$[$me$][..], and
                 outbuf$_f$[$me$][..] from worker $w$;

                 Remove $w$ from $busylist$;
             }
         }
     }

     Compress and store received edges to the underlying DB;
  }

\end{algorithm}

\section{Spatio-temporal Graph Query\\Processing and Optimization}
\label{sec:optimize} 

\begin{table*}[t]

\setlength{\tabcolsep}{4pt}

\begin{minipage}[t]{0.55\textwidth}
\caption{Terms used in the cost model.}
\vspace{-0.258in}
\begin{center}
\renewcommand{\arraystretch}{1.128}
\small
\begin{tabular}{|l|l|l|l|}
\hline
$S_e  $   & edge size                          & $[t_s,t_e]$& time range of a query\\\hline
$R_e  $   & \#edges ingested per second        & $N_q  $   & \#regions or \#slots accessed by query \\\hline
$N    $   & \#regions or \#slots in total      & $P_q  $   & $P_q=N_q/N$ \\\hline
$c_r  $   & cost for reading from disk         & $f    $   & network communication factor     \\\hline
$c_a  $   & cost for invoking backend          & $p    $   & degree of parallelism \\ \hline
$\beta$   & penalty factor $\beta \in$ (0,1)   & $S_T  $   & edge data size in a TRU: $S_eR_e\cdot$ TRU \\ \hline
$N_t  $   & \multicolumn{3}{l|}{\#TRU covered by a query's time range:
$\lceil{\tfrac{t_e}{TRU}}\rceil - \lceil{\tfrac{t_s}{TRU}}\rceil + 1$}
\\
\hline

\end{tabular} 
\label{tab:def-params}
\end{center}
\vspace{-0.25in}
\end{minipage}
\begin{minipage}[t]{0.45\textwidth}
\caption{Cost computation for Q1--Q4.}
\label{tab:cost-replica}
\vspace{-0.25in}
\small
\begin{center}
\begin{tabular}{|c|c|c|}
\hline
          &  spatio-temporal (ST) & key-temporal (KT)  \\
\hline
Q1        &  $\tfrac{c_rS_TN_t + 1048576c_a}{10\beta}$    
          &  $\tfrac{c_rS_TN_t + 16384c_a}{16384\beta}$  \\
\hline
Q2        &  $\tfrac{c_rS_TN_t + 1048576c_a}{10\beta}$    
          &  $\tfrac{c_rS_TN_t + 16384c_a}{16384\beta}$  \\
\hline
Q3        &  $\tfrac{(c_r + f_3)S_TN_t + 1048576c_a}{10\beta}$    
          &  $\tfrac{(c_r + f_3)S_TN_t + 16384c_a}{10\beta}$  \\
\cline{2-3}
KT+ST
          & \multicolumn{2}{c|}{
             $\tfrac{c_rS_TN_t + 16384c_a}{16384\beta} +
              \tfrac{(c_r+f_{sk3})S_TN_tx}{10485760\beta} + \tfrac{xc_a}{10\beta}$
            }\\
\hline
Q4        &  $\tfrac{(c_r + f_4)S_TN_t + 1048576c_a}{10\beta}$    
          &  $\tfrac{c_rS_TN_t + 16384c_a}{10\beta}$  \\
\hline
\end{tabular} 
\end{center}
\vspace{-0.25in}
\end{minipage}

\end{table*}

Spatio-temporal query processing in PAST has two distinctive features
compared to query processing in existing systems.
First, there are two edge partitions: spatio-temporal and
key-temporal.  We study the cost models for choosing edge partitions
in Section~\cref{subsec:cost-plan}.
Second, there are predicates on edge similarity, which requires
similarity joins on the spatial and/or temporal dimensions.  We
optimize the evaluation of edge similarity in
Section~\cref{subsec:edge-similarity}.

\subsection{Cost-based Partition Selection} 
\label{subsec:cost-plan}

Query processing can exploit the two types of edge partitions to skip
accessing a large amount of unrelated data and to reduce cross-machine
communication overhead.   In this subsection, we derive a cost model
for the queries.  

Table~\cref{tab:def-params} lists terms used in the model. Among the
terms, $S_e$, $R_e$, $S_T$, $c_r$, and $c_e$ are constants.  $N$ is
the total number of regions (slots) in spatio-temporal (key-temporal)
partitioning.  For example, there are 16384 slots and 1048576 regions
in our experiments.  The time range parameters $t_s$, $t_e$, and $N_t$
are determined by a given query.  $N_q$, $P_q$, $f$, and $p$ are
dependent on both the given query and the chosen edge partition type.

The data size accessed by a query can be calculated as follows:
\vspace{-0.06in}
\begin{equation*}
	data\_size = S_T N_t P_q
\vspace{-0.06in}
\end{equation*}

\noindent
Then, we formulate the total cost for evaluating the query: 
\vspace{-0.06in}
\begin{equation*}
\begin{aligned}
    Cost &= \tfrac{1}{\beta p}(Cost_{read} + Cost_{invoke} + Cost_{shuffle}) \\
         &= \tfrac{1}{\beta p}[(c_r+f)S_TN_tP_q + c_aN_q]
\end{aligned}
\vspace{-0.06in}
\end{equation*}

\noindent
where $\beta$ is a penalty factor. The actual degree
of parallelism is $\beta p$.

$Cost_{read}$ is the cost of reading $data\_size$ amount of data from
disk: $Cost_{read} = c_r data\_size$.  $Cost_{invoke}$ is the cost of
invoking the underlying DB.  PAST retrieves all $N_t$ sub-partitions
of every region / slot with a single invocation.  Therefore, the
number of invocations is $N_q$, and $Cost_{invoke} = c_a N_q$.
$Cost_{shuffle}$ is the cost of communicating a fraction of the
retrieved data across machines.  $Cost_{shuffle} = f\cdot data\_size$.
Here, parameter $f$ depends on the query evaluation strategy and the
network bandwidth.

Given the cost model, we compute the cost of processing Q1--Q4 using
different partitions.  Table~\cref{tab:cost-replica} summarizes the
computed costs.  We set $N$ to be 1048576 and 16384 for
spatio-temporal and key-temporal partitions, respectively.  The degree
of parallelism ($p$), the network communication factor ($f$) and
region/slot proportion ($P_q$) differ as different partitions are
selected for query execution.  There are 10 workers in our
experiments.  Therefore, $p \leq 10$.

\Paragraph{Q1: object trace} 
(i) \emph{spatio-temporal:} Since the given object $o$ may visit any
spatial locations during the time range, every worker reads its own
spatio-temporal partition to look for edges that contain object $o$.
Every region needs to be examined.  No edge data shuffle is necessary.
So $p=10, f=0, N_q=1048576, P_q=1$.

\noindent(ii) \emph{key-temporal:} Only the worker that contains the
key-temporal partition of $o$ reads the $o$'s slot.
There is no data shuffling during computation. 
So $p=1, f=0, N_q=1, P_q=\tfrac{1}{16384}$.

\Paragraph{Q2: trace similarity} 
(i) \emph{spatio-temporal:} Compared with Q1, the query processing
strategy is unchanged.  Now every worker looks for edges that contain
\emph{either given objects}.  So the parameters are also unchanged.
$p=10, f=0, N_q=1048576, P_q=1$.

\noindent (ii) \emph{key-temporal:} We perform Q1 then computed
similarity between the obtained traces. In the worst case, the
key-temporal partitions of the two given objects reside in two
machines.  So $p=2, f=0, N_q=2, P_q=\tfrac{1}{8192}$.

\Paragraph{Q3: similar object discovery} 
Q3 is a heavy-weight query. Both the spatio-temporal and the
key-temporal solutions read all data in the time range.  Then they
perform a join and a groupby operation by shuffling all the retrieved
data among workers.  Therefore, the network communication factor
($f_3$) is very large.  (i)  \emph{spatio-temporal:} $p=10,
N_q=1048576, P_q=1$.  (ii) \emph{key-temporal:} $p=10, N_q=16384,
P_q=1$.

We design an optimized execution strategy for Q3 that combines the two
partitions.  (iii) \emph{key-temporal+spatio-temporal:} It first
obtains the trace of the given object $o$ by Q1 accessing the
key-temporal partitions.  Then it finds all locations visited by $o$,
computes the regions from the locations, and sends the locations to
workers that store the relevant regions.  After that, each worker
reads the relevant regions from its spatial-temporal partition and
looks for edges that are similar to $o$'s trace.  Finally, similar
edges are grouped by objects to compute the aggregate similarities,
which are then sorted to obtain the query result.  While the final
step performs data shuffling, the amount of data shuffled is
reasonably small compared to that in (i) and (ii).  Therefore, the
network communication factor $f_{sk3} \ll f_3$.  Suppose the number of
regions to access is $x$ ($x \ll 1048576$), and all workers
participate in the computation.  So $p=10, N_q=x,
P_q=\tfrac{x}{1048576}$.

\Paragraph{Q4: clone object detection}
(i)  \emph{spatio-temporal:} All machines participate in computation.
They have to shuffle all the data in the time range.  Suppose the
network communication factor is $f_4$. Then $f_4 \approx f_3$ for the
same time range.  Other parameters are $p=10, N_q=1048576, P_q=1$.

\noindent(ii) \emph{key-temporal:} Each machine reads every slot in
the time range.  Since all edges of an object are in the same slot,
velocity computation can be performed locally without data shuffling.
Thus, $f=0, p=10, N_q=16384, P_q=1$.

\Paragraph{Comparison} From Table~\cref{tab:cost-replica}, we see that
for Q1, Q2, and Q4, KT's cost is lower than ST's cost: $Cost_{KT} <
Cost_{ST}$.  Therefore, KT is the better partition to use.  For Q3,
$Cost_{ST+KT} < Cost_{KT} < Cost_{ST}$.  Therefore,  ST+KT should be
selected.

In general, for a given query, we can apply the above analysis to
compute the cost for different partition types, and choose the best
partition or partition combination based on the computed costs.


\subsection{Edge Similarity Computation}
\label{subsec:edge-similarity}

\Paragraph{Optimizing Location Computation} We would like to improve
the efficiency of computing $dist(l_1, l_2) \leq TH_{dist}$.  The
basic idea is to filter out far-away locations that cannot satisfy the
inequality without computing the distance.

Given a set of locations in an area (e.g., a region or a unit). We
apply a grid to the area, where a grid cell is a $\eta \times \eta$
square and $\eta \ll $ region width $a$.  It is easy to convert the
coordinates $(x,y)$ of a location $l$ into grid coordinates $(x_{g},
y_{g})$: $x_{g}=\lfloor x/\eta \rfloor$, $y_{g}=\lfloor y/\eta
\rfloor$.  For $l_1(x_1,y_1)$ and $l_2(x_2,y_2)$, their grid
coordinates are $(x_{g1},y_{g1})$ and $(x_{g2},y_{g2})$, respectively.


Let $d_x=|x_2-x_1|$, $d_y=|y_2-y_1|$, $d_{gx}=|x_{g2}-x_{g1}|$, and
$d_{gy}=|y_{g2}-y_{g1}|$.  Therefore, $d_x \in ((d_{gx}-1)\eta,
(d_{gx}+1)\eta)$ and $d_y \in ((d_{gy}-1)\eta, (d_{gy}+1)\eta)$.
According to triangle inequality, we have the following:
\vspace{-0.05in}
\begin{equation*} 
|d_x - d_y| <= dist(l_1,l_2) <= d_x+d_y 
\vspace{-0.05in}
\end{equation*}
Thus
\vspace{-0.05in}
\begin{equation*} 
(|d_{gx} - d_{gy}|-2)\eta <= dist(l_1,l_2) <= (d_{gx}+d_{gy}+2)\eta
\vspace{-0.05in}
\end{equation*}

We use the above inequalities to calculate the lower bound of the
distance between two locations.  If the lower bound is larger than
$TH_{dist}$, then we can avoid computing the actual distance, thereby
reducing computation overhead.

\Paragraph{Optimizing Time Computation}
A sub-partition contains events in a TRU.  For $|t_2-t_1| \leq
TH_{time}$ computation, we can use the sub-partition row key (which
encode the time range) to reduce data access and computation overhead.
Let $d_T = \lceil TH_{time}/\mbox{TRU} \rceil$.  If $t_1$ is in time
range $tr_1$, we only need to consider $2d_T+1$ sub-partitions with
time ranges in $[tr_1-d_T,tr_1+d_T]$.  In particular, when $TH_{time}
<$ TRU, we only need to consider three sub-partitions with time ranges
$tr_1-1$, $tr_1$, and $tr_1+1$.

\section{Experimental Evaluation}
\label{sec:evaluate}

In this section, we evaluate the performance of PAST.  We would like
to answer the following questions in the experiments:
\begin{list}{\labelitemi}{\setlength{\leftmargin}{5mm}\setlength{\itemindent}{0mm}\setlength{\topsep}{0.5mm}\setlength{\itemsep}{0mm}\setlength{\parsep}{0.5mm}}
\item Can PAST efficiently support high-throughput edge insertions
while achieving the desired graph partitions?
%
%
\item What are the benefits of the proposed techniques, such as the
diversified partition strategy, cost-based edge partition selection,
compressed edge storage, and computation optimization?
%
%
%
\item How does PAST compare to state-of-the-art systems (e.g.,
JanusGraph, Greenplum, Spark, and ST-Hadoop)?
%
%
\end{list}

\subsection{Experimental Setup}

\Paragraph{Machine Configuration} The experiments are performed on a
cluster of 11 machines.  Each machine is a Dell PowerEdge blade server
equipped with two Intel(R) Xeon(R) E5-2650 v3 2.30 GHz CPUs (10
cores/20 threads, 25.6MB cache), 128GB memory, a 1TB HDD and a 180GB
SSD, running 64-bit Ubuntu 16.04 LTS with 4.4.0-112-generic Linux
kernel.  The blade servers are connected through a 10Gbps Ethernet.
We use Oracle Java 1.8, Cassandra 2.1.9, JanusGraph 0.2.1, Greenplum
5.9.0, Spark 2.2.0, ST-Hadoop 1.2 and Hadoop 2.8.1 in our experiments.

\Paragraph{Workload Generation} 
We generate a synthetic data set of customer shopping events.  The
goal to support 10 billion object vertices and 10 million location
vertices are designed for clusters with 1000 machines.  Given the
machine cluster size in our experiments, we scale down the number of
vertices by a factor of 100.  Therefore, we would like to generate 100
million object vertices and 100 thousand location vertices.

First, we crawled 450 thousand hotel locations in China from ctrip.com
(a popular travel booking web site in China).  The locations are
distributed across about 1100 areas (cities / counties / districts).
We randomly choose 100 thousand locations from the real-world hotel
locations as shopping locations.
Second, we generate 100 million customers.  As the number of shopping
locations and the population of an area are often correlated, we
assign customers to areas so that the number of customers is
proportional to the number of locations in an area.  
Third, we would like to generate a data set that covers a 2-3 year
period and can be stored in the cluster in the experiments.  We assume
40\% people are frequent shoppers and 60\% people are infrequent
shoppers.  A frequent shopper and an infrequent shopper visit a
randomly chosen shopping location in his or her area every week with
probability 0.8 and 0.2, respectively.  Edge timestamps are used to
compute edge partitions and evaluate queries in our experiments.
(Note that the week period is necessary to reduce the total data
volume to fit into the cluster storage capacity.  Our ingestion
experiments will send the edge data as fast as possible to saturate
the system disregarding edge timestamps).
Finally, we produce a small number of cloned objects that visit
shopping locations in far-away areas.  The resulting graph contains 57
billion edges, which cover 800 days.

\Paragraph{Parameter Settings}  In spatio-temporal partitioning, we
apply a $1024 \times 1024$ grid to the map, obtaining 1048576 regions.
Note that there are more hotels in the city than in the countryside.
We see that hotels concentrate in 8660 regions.  In key-temporal
partitioning, we generate 16384 slots.  We set TRU to be 24 hours.

We evaluate Q1-Q4 as described in Section \cref{subsec:query}. By
default, we set $TH_{time} = 7h$, $TH_{dist}=100m$,
$TH_{velocity}=120km/h$.

\Paragraph{State-of-the-art Systems to Compare} We evaluate our
proposed solution, PAST, and four state-of-the-art systems in the
experiments:
\begin{list}{\labelitemi}{\setlength{\leftmargin}{5mm}\setlength{\itemindent}{0mm}\setlength{\topsep}{0.5mm}\setlength{\itemsep}{0mm}\setlength{\parsep}{0.5mm}}
\item \emph{PAST (our proposed solution)}.  We implement PAST in Java.
The GraphStore uses Cassandra as the underlying DB backend and stores
data on the HDDs.  We customize Cassandra's partitioner to manage the
key-to-node mapping.  
\item \emph{JanusGraph (a state-of-the-art distributed graph store)}.
A spatio-temporal graph is stored as a property graph. We choose
Cassandra as JanusGraph's storage backend.  To facilitate edge
retrieval for a given object (location) vertex, we set the vertex ID
to be the object ID (location ID).
\item \emph{Greenplum (a state-of-the-art MPP relational DB)}. Graph
data is stored as relational tables in Greenplum. To minimize disk
space usage, we create two tables for edge data and for location
details with latitude and longitude, respectively.  The two tables are
linked by location ID.  Q2, Q3 and Q4 perform join operations.  We
employ multi-dimensional partitioning for improving query performance.
Vertex ID is the first dimension, and time is the second dimension.  
\item \emph{Spark+Cassandra (a state-of-the-art big-data analytics
system)}. Graph data is stored in Cassandra.  Spark loads the location
data into memory at the beginning of execution to reduce the overhead
of looking up locations.  The loading time is less than 1 second, and
is negligible compared with the query execution time.  Spark accesses
data in Cassandra for query processing, and saves the query results to
HDFS.
\item \emph{ST-Hadoop+Spark (a big data system specially optimized for
spatio-temporal data)}. Given the TRU setting in PAST, we set
ST-Hadoop's partition granularity to be day.  PAST's spatio-temporal
partitioning employs the Z-order curve.  Therefore, we set ST-Hadoop's
spatial index technique to be the Z-order curve.  To avoid MapReduce's
overhead of storing intermediate data to disks, we use Spark to read
spatio-temporally indexed data in ST-Hadoop and compute the queries in
memory.
\end{list} 

\begin{figure*}[t]
  \vspace{-0.1in}
  \centering
  \subfloat
           {\includegraphics[width=7in]{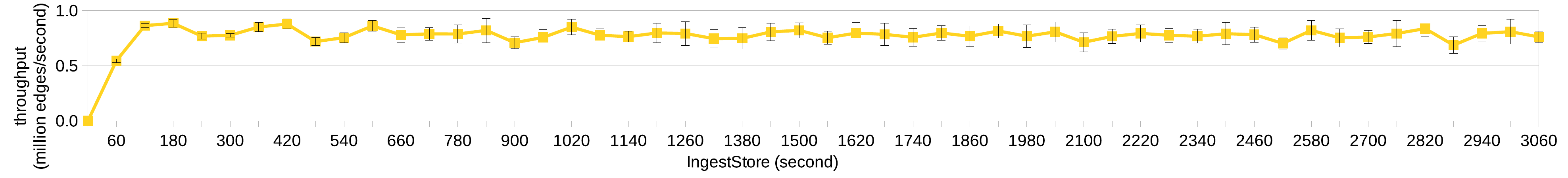}
           \label{fig:is}}
  \hfil
  \vspace{-0.1in}
  \subfloat
           {\includegraphics[width=7in]{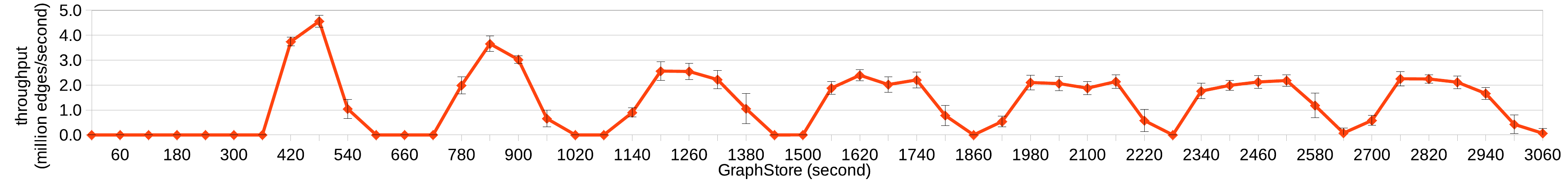}
           \label{fig:gs}}
  \vspace{-0.1in}
  \caption{Ingestion throughput seen at IngestStore and GraphStore
evolving over time
(\#workers=10).}
  \label{fig:throughput-over-time}
  \vspace{-0.1in}
\end{figure*}

\subsection{New Edge Ingestion}
\label{subsec:exp-ingestion}

In this section, we measure the sustained edge ingestion throughput 
in PAST, and examine the distribution of data across workers.

As described in Section~\cref{sec:stream}, PAST shuffles $m$-TRU worth
of data in every round.  In our experiments, we set $m$ to be the
number of worker machines.  Note that edges are partitioned during
ingestion.

\begin{figure}[t]
  \vspace{-0.1in}
  \centering
  \includegraphics[width=3in]{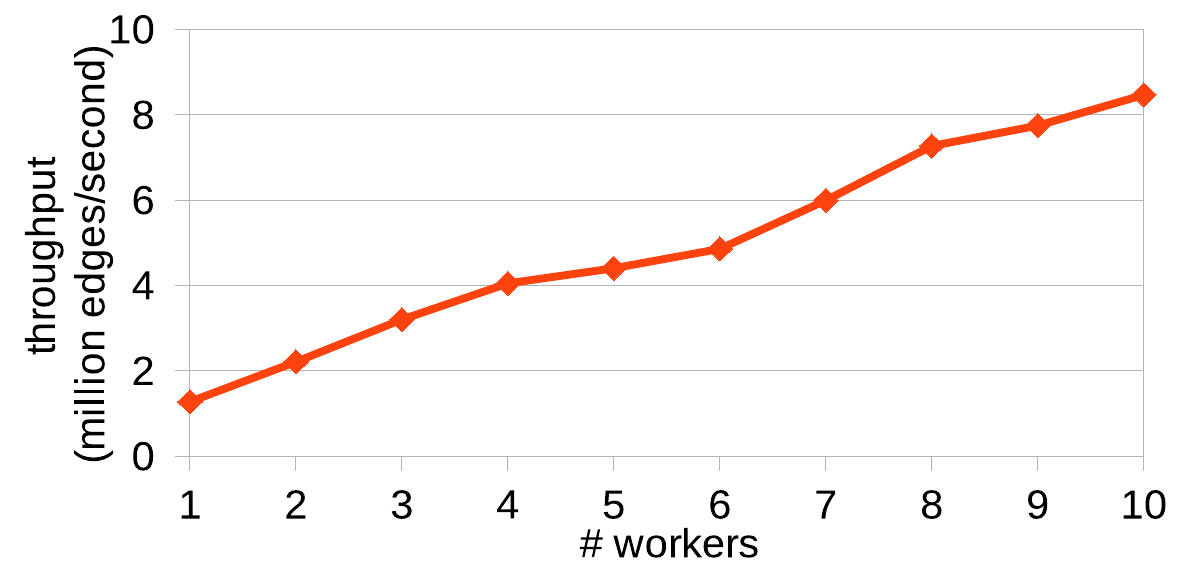}
  \vspace{-0.15in}
  \caption{Ingestion throughput varying the number of workers.}
  \label{fig:throughput}
\vspace{-0.2in}
\end{figure}

\Paragraph{Edge Ingestion Throughput and Scalability}
The edge ingestion throughput is the number of new edges streamed into
all the GraphStores per second. We send edges in the data set as fast
as possible in this set of experiments.  We observe that the ingestion
throughput stabilizes in the third round, as shown in
Fig.\cref{fig:throughput-over-time}.  Given this, we measure the
ingestion throughput in the fifth and sixth rounds at IngestStores,
and report the aggregate ingestion throughput as the sustained
throughput in Fig.\cref{fig:throughput}.

In Fig.\cref{fig:throughput}, we vary the number of worker machines on
the X-axis from 1 to 10 to study the scalability of our solution.  The
Y-axis reports the sustained ingestion throughput in million edges per
second.  From the figure, we see that the sustained ingestion
throughput increases nearly linearly as the number of worker grows.
The PAST design achieves good scalability for new edge ingestion.

Every worker in PAST can support an additional 0.85 million new edges
per second.  An edge takes 24 bytes in this experiment.  Every worker
in PAST can support additional 20MB/s ingestion bandwidth.  Therefore,
the design goal of 1 trillion new edges per day (or 290MB/s) for a
full-scale spatio-temporal graph can be achieved with about 15 worker
machines.  This gives a lower bound of the actual number of nodes in a
design, whose choice must also consider the performance of query
processing.

\Paragraph{Edge Ingestion Throughput Evolving Over Time} We measure
the number of ingested edges seen by both IngestStores and GraphStores
at every minute for about 50 minutes.
Fig.\cref{fig:throughput-over-time} shows the average per-machine
throughput across all machines.  The X-axis is wall-clock time.  The
Y-axis is the ingestion throughput in million edges per second.  The
error bars show the standard errors.

The upper figure in Fig.\cref{fig:throughput-over-time} shows the
ingestion throughput seen by IngestStores.  IngestStores begin to
handle incoming edges at time 0. The ingestion throughput quickly
increases at the beginning.  Then it fluctuates around 0.85 million
new edges per second.

The lower figure in Fig.\cref{fig:throughput-over-time} shows the
ingestion throughput seen by GraphStores.  GraphStores begin to
receive data at time=360s. Its throughput reaches the peak with nearly
108MB/s (4.5 million edges per second) at time=480s.  This is because
Cassandra starts with empty in-memory buffers for receiving data, and
the cost of storing to Cassandra is very low at the beginning.  As
time goes by, Cassandra moves into more steady states.  The peak
throughput reduces to 86MB/s (3.6 million edges per second) at the
second round.  The peak throughput in stable rounds (i.e. round 3 and
beyond) is about 48MB/s (2.0 million events/s).

Note that IngestStores send both key-temporal data and spatio-temporal
data.  The shuffled data is essentially twice the amount of the
incoming edge data.  Therefore, the ingestion throughput seen at
GraphStores is roughly twice as much as that seen at IngestStores.

\Paragraph{Data Distribution} We would like to understand how well our
partitioning methods work in terms of data distribution across
machines.  We measure the data size in each worker for both
spatio-temporal edge partitions and key-temporal edge partitions.
In Fig.\cref{fig:data_amount}, the X-axis shows each worker machine,
while the Y-axis reports data size in GB.  From the figure, we see
that PAST's partitioning methods achieve balanced data distribution.
Note that the data size of spatio-temporal partitions is roughly twice
as large as the data size of key-temporal partitions.  This is because
PAST stores three replicas for edge data, of which two in
spatio-temporal partitions, and one in key-temporal partitions.

\begin{figure}[t]
  \vspace{-0.1in}
  \centering
  \includegraphics[width=3.25in]{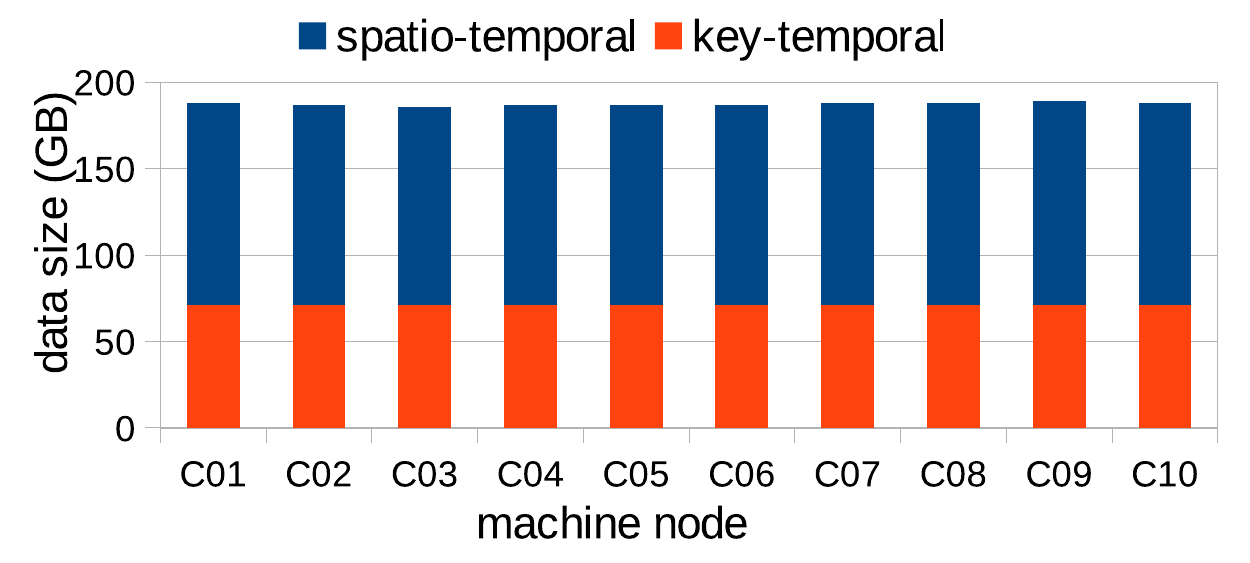}
  \vspace{-0.15in}
  \caption{Data distribution in different edge partitions.}
  \label{fig:data_amount}
\vspace{-0.20in}
\end{figure}

\subsection{Proposed Features in PAST}
\label{subsec:exp-past}

In this section, we evaluate the proposed features in PAST, including
storage format, edge partition selection, computation optimization,
and the spatial mapping algorithms.

\subsubsection{Comparison of Two Storage Formats}
\label{subsec:exp-storage}

We evaluate the two storage formats as described in
Section~\cref{subsec:store}: (i) $C$: this is the baseline design,
where edge properties are stored in separate tables (column families)
with compression; and (ii) $R$: this is the PAX design, where the
compressed columns are concatenated to store in a single table (column
family).  Since they both store compressed column data, $C$ and $R$
consume the same amount of disk space.  However, their query
performance can be different.

\begin{figure}[h]
  \vspace{-0.1in}
  \centering
  \includegraphics[width=3.2in]{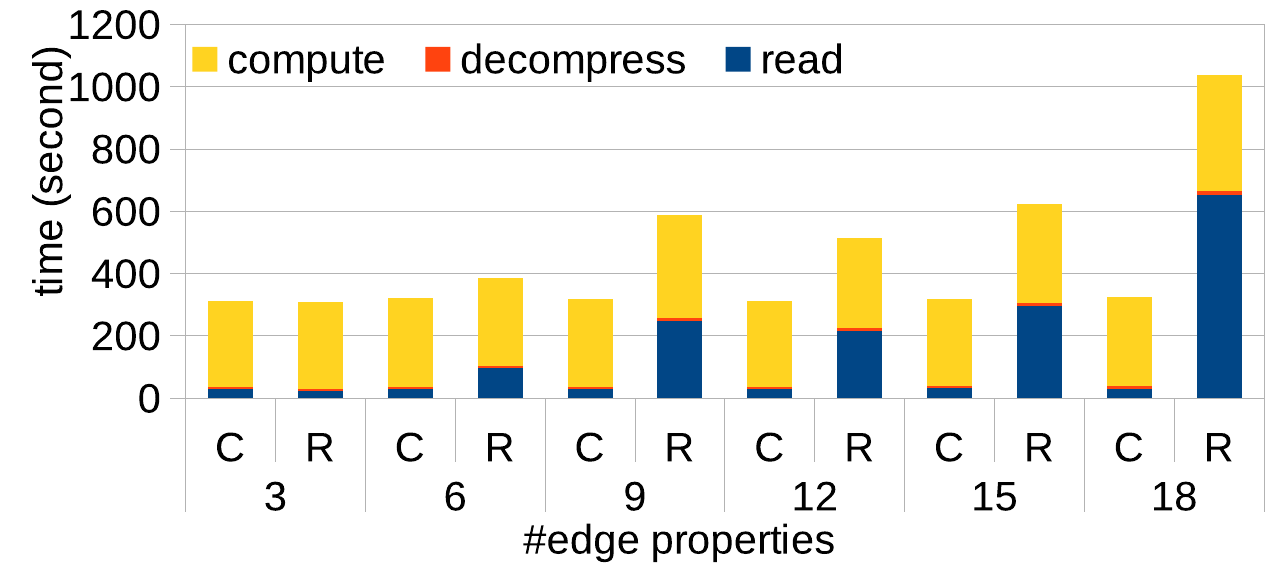}
  \vspace{-0.15in}
  \caption{Effect of storage format for Q4.}
  \label{fig:effect-storage}
\vspace{-0.1in}
\end{figure}

Fig.\cref{fig:effect-storage} reports query execution time for Q4 while
varying the number of edge properties.  Note that Q4 uses only three
edge properties (i.e., object ID, timestamp, location ID) in
computation.  We perform the experiment on a single machine.  The
X-axis shows the total number of edge properties, and the Y-axis is
elapsed time.  

From the figure, we see that the cost of decompression is negligible.
When there are three edge properties, $C$ and $R$ take nearly the same
time for reading data from Cassandra because they retrieve the same
amount of data from disk. However, as the number of properties
increases beyond three, the query performance of $R$ deteriorates.
The read time in $R$ increases as the number of properties increases.
This is because $C$ avoids accessing irrelevant properties, while $R$
has to retrieve all properties.

\subsubsection{Benefits of Cost-Based Edge Partition Selection}
\label{exp:effect-replica}

We perform a set of experiments to validate the analysis in
Section~\cref{subsec:cost-plan}.  Here, we evaluate Q1, which
represents simple graph traversal queries, and Q3, which represents
complex queries, using different edge partitions.  The query time
range is one day.

Table~\cref{tab:effect-replica} lists the execution time and accessed
data size for Q1 and Q3 using the following solutions. (i) ST: a query
accesses only spatio-temporal partitions.  (ii) KT: a query accesses
only key-temporal partitions.  (iii) KT+ST for Q3: the optimized
execution strategy for Q3, where PAST first accesses key-temporal
partitions for retrieving the trace of the given object $o$, then
accesses spatio-temporal partitions for computing objects that are
similar to $o$.

From Table~\cref{tab:effect-replica}, we see that for Q1, $Cost_{KT} <
Cost_{ST}$, and for Q3, $Cost_{ST+KT} < Cost_{KT} < Cost_{ST}$.  This
observation is in accordance with the analysis based on the cost model
in Section~\cref{subsec:cost-plan}.  Overall, the best execution plans
for Q1 and Q3 achieve a factor of 39.5x and 9.7x improvements over the
second best plans.

\begin{table}[h]
  \centering
\caption{Query performance using different edge partitions.}
\vspace{-0.1in}
\small
\begin{tabular}{|c|cc|cc|}
\hline
           & \multicolumn{2}{c|}{Q1}  &  \multicolumn{2}{c|}{Q3}       \\
\cline{2-3} \cline{4-5}
replica    & size (MB) & time (s)    & size (MB) & time (s)          \\
\hline
KT         &   0.058   &  \textbf{0.47}   &  950      &   343.92          \\
ST         &   950     &  18.57           &  950      &   636.52          \\
KT+ST      &   --      &  --              &  9        &   \textbf{35.63}  \\
\hline
\multicolumn{5}{l}{Note: only Q3 can be optimized using both two edge partitions.}
\end{tabular}
\label{tab:effect-replica}
\vspace{-0.15in}
\end{table}

\subsubsection{Benefits of Computation Optimization}
\label{subsec:exp-opt}

We study the effect of optimization techniques for evaluating edge
similarity, as described in Section~\cref{subsec:edge-similarity}.  We
measure Q3's execution time while varying the query time range from 1
day to 512 days.  We compare three solutions: (i) ST-OP: the baseline
implementation in PAST with both location and time computation
optimized; (ii) T-OP: only time computation is optimized; (iii) NOP:
there is no computation optimization.

\begin{table}[h]
\vspace{-0.1in}
\caption{Query results varying time range for Q3.}
\vspace{-0.25in}

\setlength{\tabcolsep}{3pt}

\begin{center}
\small
\begin{tabular}{|c|c|c|c|c|c|c|c|c|c|c|}
\hline
\#day 	           & 1&	2&	4&	8	&16	&32	&64	&128&	256	&512 \\
\hline
read(MB)           &9 &34 &34 & 89 &221 &456 &963 & 2000 & 4005& 8134 \\
\hline
\#result($\times10^3$) & 3&	7&	7&	13	&27	&45	&87	&174&	333	&566 \\
\hline                                               
\end{tabular} \label{tab:effect-algo}
\end{center}
\vspace{-0.3in}
\end{table}

\begin{figure}[h]
  \vspace{-0.1in}
  \centering
  \includegraphics[width=2.7in]{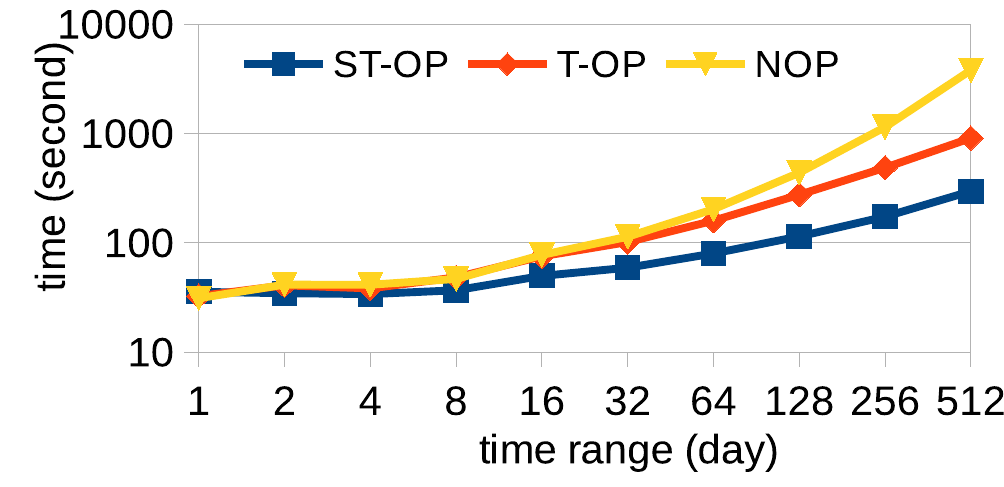}
  \vspace{-0.15in}
  \caption{Effect of algorithm optimizations for Q3.}
  \label{fig:effect-algo}
  \vspace{-0.1in}
\end{figure}

Fig.\cref{fig:effect-algo} reports query execution time in the
logarithmic scale.  Table~\cref{tab:effect-algo} lists the read data
size and the number of query results for all the experiments in
Fig.\cref{fig:effect-algo}.  When the query time range is less than 4
days, the three solutions take nearly the same time.  As the query
time range increases, the performance difference becomes obvious.
Overall, ST-OP outperforms T-OP and N-OP by a factor of up to 3.1x and
12.8x.  The computation optimization is quite significant.

\subsubsection{Effect of Spatial Mapping for Skewed Data}
\label{subsec:exp-skew}

\begin{figure}[t]
 \vspace{-0.1in}
  \centering
  \subfloat[Workload distribution]
           {\includegraphics[width=2in]{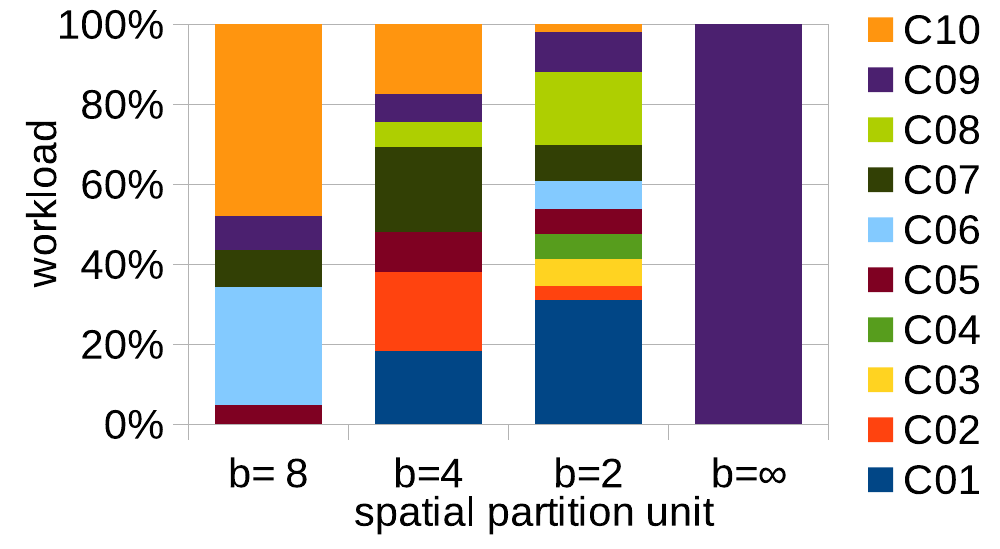}
           \label{fig:dist}}
  \hfil
  \subfloat[Runtime]
           {\includegraphics[width=1in]{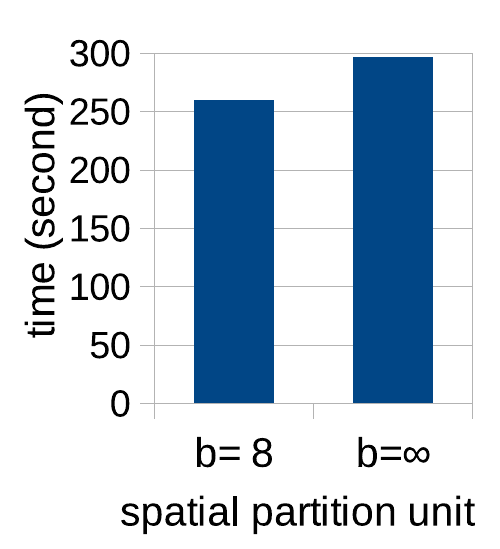}
           \label{fig:runtime}}
  \vspace{-0.1in}
  \caption{Effect of skew spatial data partitioning for Q3.}
  \label{fig:effect-skew}
  \vspace{-0.1in}
\end{figure}

We compare unbounded spatial mapping (Alg.\cref{algo:sp-unbounded}) and
bounded spatial mapping (Alg.\cref{algo:sp-bounded}) using Q3.  We
consider three choices of b (i.e., 8, 4, and 2) for
Alg.\cref{algo:sp-bounded}.  According to Theorem~2, the corresponding
lower bound of SF (spatial factor) is 88\%, 77\% and 56\%,
respectively.  (we set region width $a$ to be $D$).  The unbounded
algorithm is labeled by $b=\infty$.  PAST employs the optimized
execution plan for Q3 as described in Section~\cref{subsec:cost-plan}.

Fig.\cref{fig:dist} shows the distribution of accessed regions for each
machine when the query time range is 512 days.  The smaller the SF is,
the more machines are likely to contain regions used in the query.
However, we would like to choose $b$ such that $b$ is not very small
and data locality is maintained for most location computation.
Fig.\cref{fig:runtime} shows the execution time for $b=8$ and
$b=\infty$.  We see that the bounded algorithm achieves a small
improvement of 10\% over the unbounded algorithm.

\subsection{Comparison with State-of-the-art Systems}
\label{subsec:exp-compare}

We compare PAST with state-of-the-art systems, including JanusGraph,
Greenplum, Cassandra+Spark, and ST-Hadoop+Spark.  We are interested in
two aspects: (i) storage space consumption and (ii) query performance.

\subsubsection{Storage Space}
\label{subsec:exp-space}

Table~\cref{tab:comparison-data-size} shows storage space used in all
systems\footnote{\small Data ingestion times for the systems are as
follows.  PAST: 6h, JanusGraph: 20h, Greenplum: 16h, Cassandra: 46h,
ST-Hadoop: 12h.}.  Note that PAST and Cassandra keep 3 replicas, and
Greenplum stores 2 replicas by default.  Due to disk space limitation,
JanusGraph, Greenplum, and ST-Hadoop stores only 1 replica.

\begin{table}[h]
\vspace{-0.05in}
\caption{Disk Space Consumption}
\vspace{-0.25in}

\setlength{\tabcolsep}{3pt}

\begin{center}
\small
\begin{tabular}{|c|c|c|c|c|c|}
\hline
          & PAST & JanusGraph & Greenplum & Cassandra & ST-Hadoop \\
\hline
size (TB) & 1.9  &  3.4       & 2.8       &  3.2      & 3.1 \\
\hline 
\#replica & 3    &  1         & 2         &  3        & 1   \\
\hline 
\end{tabular} \label{tab:comparison-data-size}
\end{center}
\vspace{-0.20in}
\end{table}

\begin{figure*}
\centering
$
\begin{array}{c}

  \includegraphics[width=6in] {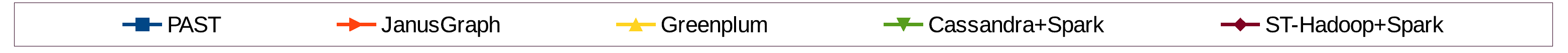} \\

  \begin{array}{cccc}

  \includegraphics[width=1.6in] {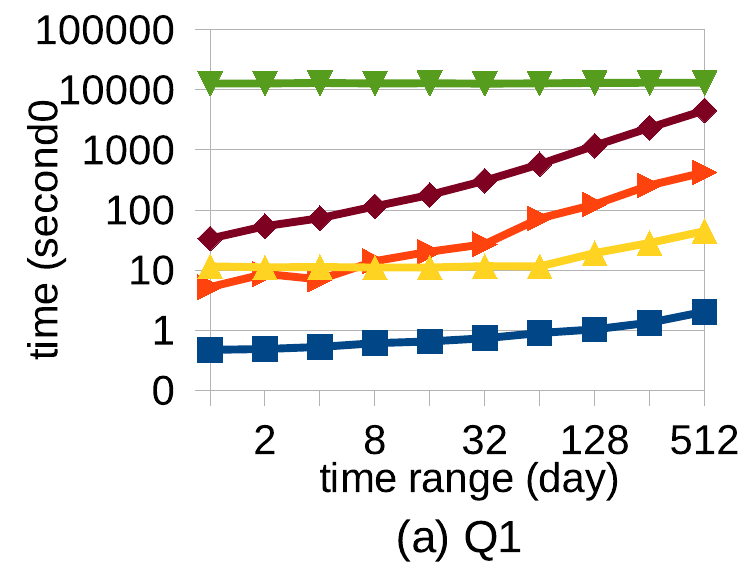}  &

  \includegraphics[width=1.6in]{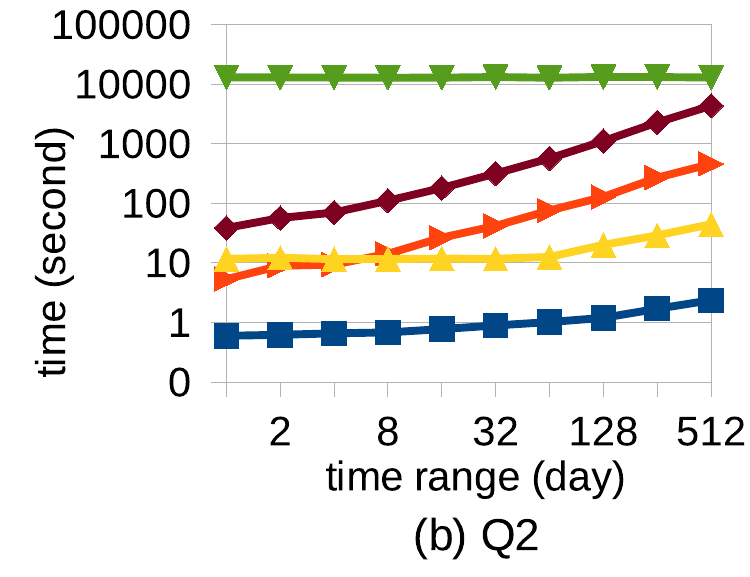}   &

  \includegraphics[width=1.6in]{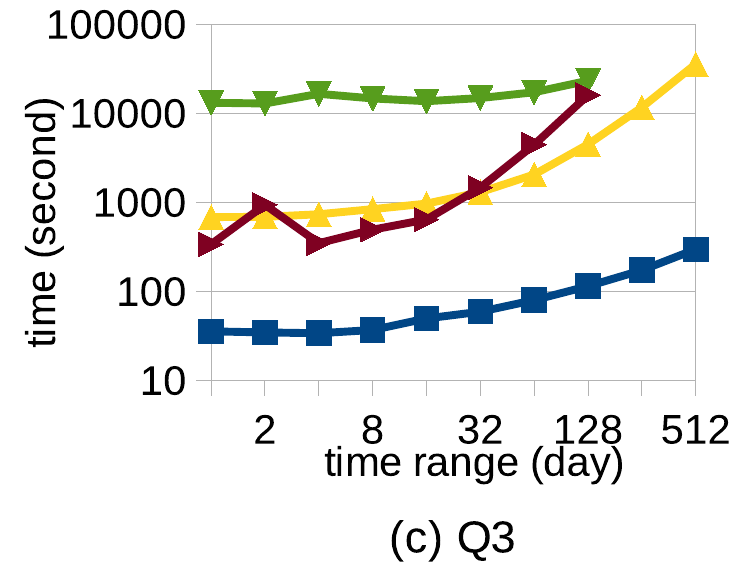}   &
	
  \includegraphics[width=1.6in]{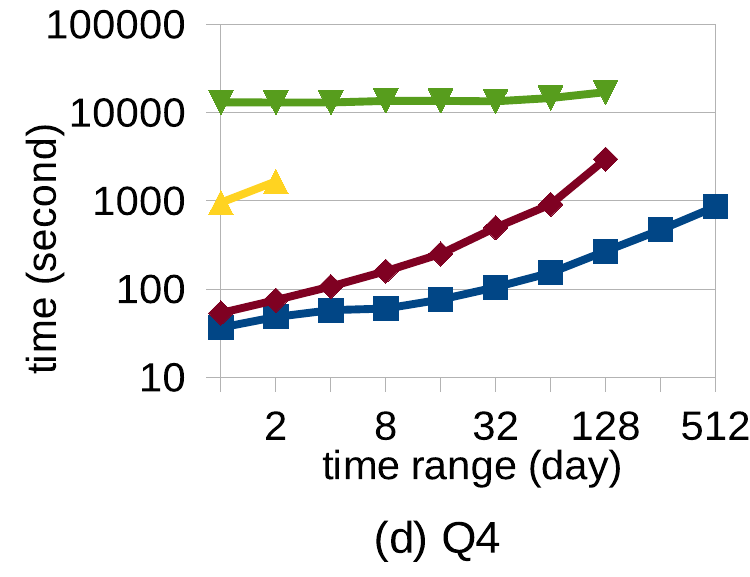} 

  \end{array}
\end{array}
$
  \vspace{-0.2in}

  \caption{Query performance comparison ($TH_{time}$ = 7h, $TH_{dist}$ = 100m). }
  \label{fig:comparison-sel}

  \vspace{-0.2in}

\end{figure*}

From the table, we see that compared with Cassandra, PAST achieves a
factor of 1.7x space savings.   PAST consumes less space than the
other systems even if they store fewer number of replicas.  If we
calculate the space consumption for JanusGraph, GreenPlum, and
ST-Hadoop to account for three replicas, then PAST reduces the space
consumption of JanusGraph, GreenPlum, and ST-Hadoop by a factor of
5.4x, 2.2x, and 4.9x, respectively.

Note that JanusGraph stores the graph data in Cassandra.  Both
Cassandra and JanusGraph employ LZ4 compression.  In comparison, PAST
achieves much better space consumption.  This is because PAST utilizes
columnar layout for all edges in a sub-partition, and compresses the
columnar edge properties.  JanusGraph consumes more space than
Cassandra because it stores each edge twice at both the incoming
vertex and the outgoing vertex.  PAST and Cassandra take less space to
store one replica than Greenplum because of compression.

\subsubsection{Query Performance}
\label{subsec:exp-query}

Fig.\cref{fig:comparison-sel} and Fig.\cref{fig:comparison-sel2} compare
the query performance for all systems, while varying the query time
range and the threshold values, i.e. $TH_{time}$ and $TH_{dist}$.  The
Y-axis is executing time in the logarithmic scale.  We do not run Q3
and Q4 on JanusGraph as it mainly focuses on simple traversal queries
and employs Spark for complex queries.  Therefore, Q3 and Q4 on
JanusGraph can be represented by Q3 and Q4 on Cassandra+Spark,
respectively.

There are several missing points in Fig.\cref{fig:comparison-sel}(c)
and (d).  The experiments corresponding to the missing points run over
one day and have not completed.  Cassandra+Spark and ST-Hadoop+Spark
are overwhelmed by shuffling for Q3 and Q4 with large query time
ranges.   For Q4, most systems compute the velocity of edges of an
object in sorted time order.   In contrast, in GreenPlum, the SQL
query for Q4 would read and join all data for computing velocity,
which quickly overwhelm the system.
 
\begin{figure}
\centering
$
\begin{array}{c}

  \includegraphics[width=3.0in] {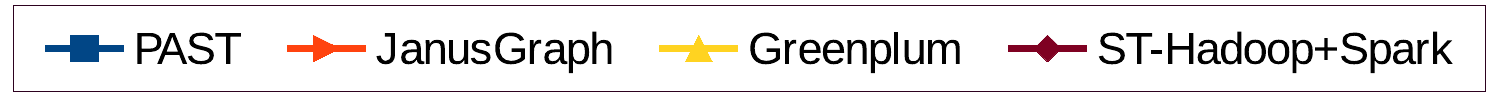} \\

  \begin{array}{cc}

  \includegraphics[width=1.4in] {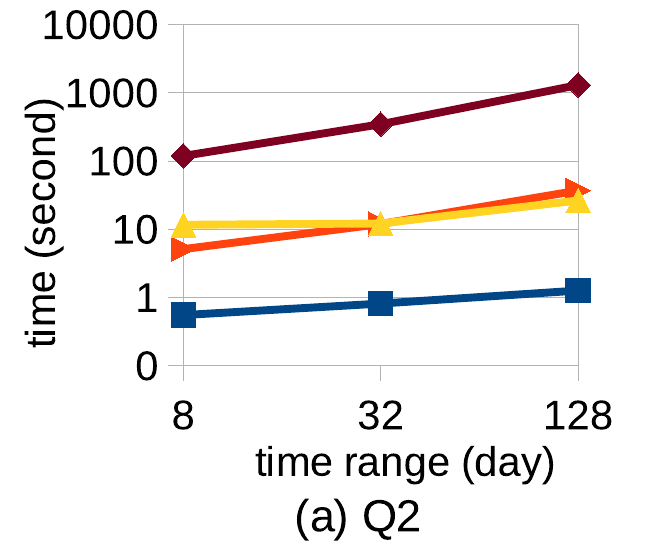}  &

  \includegraphics[width=1.4in]{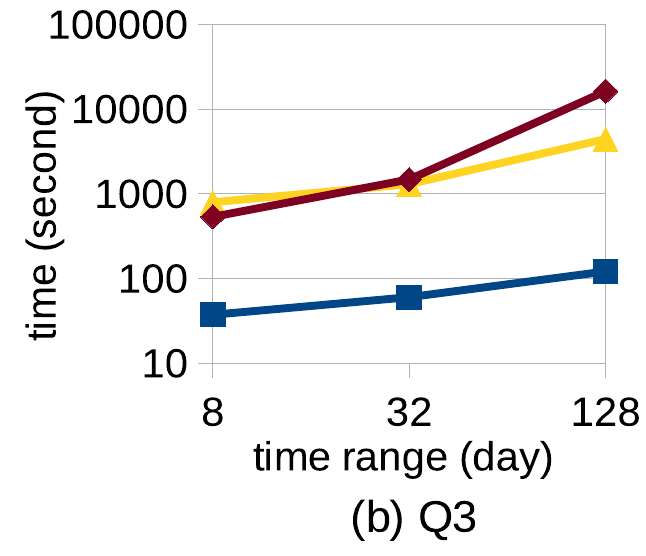}   \\

  \end{array}
\end{array}
$
  \vspace{-0.2in}

  \caption{Query performance comparison ($TH_{time}$ = 10h, $TH_{dist}$ = 1000m).}
  \label{fig:comparison-sel2}

  \vspace{-0.2in}

\end{figure}

From Fig.\cref{fig:comparison-sel} and Fig.\cref{fig:comparison-sel2},
we see that PAST achieves 1--4 orders of magnitude better performance
compared with the four existing solutions.  The partition and query
processing schemes in PAST can effectively reduce the amount of data
accessed from the underlying storage and the data communication cost.
The main bottleneck of Cassandra+Spark is the disk I/Os for scanning
all data.
ST-Hadoop+Spark performs better than Cassandra+Spark because it
exploits the spatio-temporal index to reduce the amount of data to
read.  However, as the query time range increases, ST-Hadoop+Spark's
performance degrades and approaches that of Cassandra+Spark.
Greenplum achieves significantly better performance than
ST-Hadoop+Spark for the simple queries, Q1 and Q2. This is because
Greenplum first partitions by object then further sub-partitions by
time.  The object partitions fit the needs of Q1 and Q2 well.  In
contrast, there is no object indexing support in ST-Hadoop.

\section{Discussion}
\label{sec:discussion}

\Paragraph{Medium-grain Partitions for Object Vertices} Our partitioning scheme
for object vertices is based on hash partitioning.  However, if we know more
about the object vertices, we may design better partitioning schemes.  Note
that too fine grained partitions result in the difficulty in recording the
partitions and the significant expense in computing the partitions.  Therefore,
we consider medium-grain partitioning based on groups of vertices.  We can add
a group property to every vertex to record the group of the vertex.  Then the
partition decision is based on groups rather than every vertex. The group to
machine assignment can be captured in the coordinator node, who also keeps
track of the group statistics. 
The assignment of vertices to groups is application dependent. For example, it
would be nice if people with similar behaviors are assigned to the same group.

\Paragraph{Dynamic Vertex Group Partitions} Every worker node periodically 
reports group statistics to the coordinator machine.  Based on the 
statistics, the coordinator machine computes the average (space and 
computation) loads per worker machine.  It will ask worker nodes to migrate 
vertex groups if the load of a machine is beyond a threshold (e.g.,
+/- 5\% of the average).  The vertex groups to be migrated can be computed 
based on the collected statistics.

\Paragraph{Alternative Implementation as Index on RDBMSs} We present a
stand-alone implementation of PAST in this paper.  Alternatively, we can
implement PAST as a spatio-temporal index structure on top of RDBMSs.  Then
PAST can be applied to more applications that combine both spatio-temporal
graphs and traditional relational data.  PAST speeds up the part of queries
that involve spatio-temporal graphs and returns a list of vertex IDs or edge
IDs.  Then the list can be combined with query outputs from relational queries
to further compute the final query results.

\Paragraph{More Extended Applications} The idea of PAST can be applied to slove
more generalized multi-dimension graph problems, not limited by spatio-temporal
dimension. Moreover, the concepts of space can be generalized as user space,
commodity space, etc, and in the meanwhile, the distance can not only be
Euclidean distance but also other similarity measurement.


\section{Conclusion}\label{sec:conclusion}

In conclusion, we define a bipartite graph model for spatio-temporal graphs
based on the commonalities of representative real-world applications, i.e.,
customer behavior tracking and mining, clone-plate car detection, and shipment
tracking.   We propose and evaluate PAST, a framework for efficient
\underline{PA}rtitioning and query processing of
\underline{S}patio-\underline{T}emporal graphs.  Our experimental evaluation
shows that PAST can meet the requirements of the above applications.  Our
proposed partitioning and storage methods and algorithm optimizations achieve
significant performance improvements.  For typical queries on spatio-temporal
graphs, PAST can outperform state-of-art systems (e.g., JanusGraph, Greenplum,
Spark, ST-Hadoop) by 1--4 orders of magnitude.


%

\balance

\clearpage

\bibliographystyle{abbrv}
\bibliography{manuscript}

\end{document}